\documentclass{revtex4}
\usepackage{amsmath,amsthm}
\usepackage{rotating}
\usepackage{multirow}

\usepackage{graphicx}

\usepackage{amsfonts}
\newtheorem{theorem}{Theorem}[section]
\newtheorem{lemma}[theorem]{Lemma}
\newtheorem{proposition}[theorem]{Proposition}
\newtheorem{cor}[theorem]{Corollary}

\theoremstyle{remark}
\newtheorem{remark}[theorem]{Remark}

\theoremstyle{definition}
\newtheorem{definition}[theorem]{Definition}

\theoremstyle{example}
\newtheorem{example}[theorem]{Example}

\theoremstyle{notation}

\newcommand{\bra}[1]{\langle#1|}
\newcommand{\ket}[1]{|#1\rangle}

\begin{document}

\title{Coherent spaces, Boolean rings and quantum gates}            
\author{A. Vourdas}
\affiliation{Department of Computer Science,\\
University of Bradford, \\
Bradford BD7 1DP, UK\\}

\begin{abstract}
Coherent spaces spanned by a finite number of coherent states, are introduced. Their coherence properties are studied, using the Dirac contour representation. It is shown that the corresponding projectors resolve the identity, and that they transform into projectors of the same type,
under displacement transformations, and also under time evolution. The set of these spaces, with the logical OR and AND operations is a distributive lattice,
and  with the logical XOR and AND operations is a Boolean ring (Stone\rq{}s formalism).
Applications of this Boolean ring into classical CNOT gates with $n$-ary  variables, and also quantum CNOT gates with coherent states, are discussed.
\end{abstract}
\maketitle

\section{Introduction}

Coherent states have been studied for a long time\cite{coh1,coh2,coh3}. They are non-orthogonal states with well known properties, and they have been used extensively in the general area of quantum optics and quantum information.

We consider $n$-dimensional spaces, spanned by $n$ coherent states (a finite number of coherent states are linearly independent).
We show that these spaces have coherence properties, analogous to those of coherent states:
Firstly, there is a resolution of the identity in terms of the projectors to all these spaces (also some smaller sets of coherent spaces, are total sets).
Secondly, they have a closure property, where under both displacement transformations and time evolution, these projectors are transformed into other projectors of the same type.
Thirdly, they obey some relations which can be viewed as extensions of the fact that coherent states are eigenstates of the annihilation operator.

We prove these properties using the Dirac contour representation\cite{DI1,DI2,DI3,DI4,DI5}. 
This represents a ket state with its Bargmann function, and the corresponding bra state
with another function that has poles in the complex plane. The scalar product is then given by a contour integral. 
In this language, a coherent space is described by a finite set of poles in the complex plane (the number of poles is equal to the dimension of the space).

The set of all finite sets of complex numbers (which in our context are poles describing coherent spaces) with union  (logical OR) and intersection (logical AND), is a distributive lattice. Following Stone \cite{S,S1,JS} we describe this lattice as Boolean ring, with the symmetric difference (logical XOR) as addition, and the  intersection (logical AND) as multiplication. The properties of this Boolean ring are discussed,
and provide the theoretical foundation for both classical gates and quantum gates. 
Especially, if we go from binary variables to n-ary variables, certain statements which are trivial in the case of binary variables, become complex in the case of n-ary variables, and the Boolean ring structure is used in their proof.

These results are transferred isomorphically to coherent spaces. The set of these spaces, with disjunction (logical OR) and conjunction (logical AND) form a distributive lattice, which is a sublattice of the Birkhoff-von Neumann lattice of subspaces (which is not distributive). 
We describe this lattice as a Boolean ring, and show that it is isomorphic to the Boolean ring of finite sets of complex numbers (poles).

We use the Boolean ring structure for the study of  CNOT gates (with $n$-ary variables), and also quantum CNOT gates.
Most of the work on quantum gates\cite{G1,G2,G3,G4,G5} uses orthogonal states.
Quantum gates with `almost orthogonal\rq{} coherent states (i.e., coherent states which are far from each other) have been studied in \cite{R,R1}.
Here we study quantum gates with coherent states, taking into account the non-orthogonal nature of coherent states.
Technically this is done with a tensor $g$ and its inverse $G$, which describe the overlap between coherent states, and which appear in the calculations.

In section II we discuss briefly coherent states and the Bargmann representation, in order to define the notation.
In section III we present the Dirac contour representation, and give some technical details which are needed later.
In section IV we introduce the coherent spaces and the corresponding projectors.
In section V we present the coherence properties of the coherent projectors (resolution of the identity; closure under displacement transformations and under time evolution; relations analogous to the `eigenstate property\rq{} of coherent states).
In section VI we study coherent projectors in the Dirac contour representation, and discuss technical details of calculations in the non-orthogonal basis of coherent states, using the tensor $g$ and its inverse $G$. 

In section VII consider the set of all finite sets of complex numbers, as a distributive lattice and a Boolean ring.
This uses the general formalism of Stone  \cite{S,S1,JS}, in our own context.
We also use this formalism to study reversible classical gates, and in particular CNOT gates.
In section VIII we transfer isomorphically this formalism to the coherent subspaces, and use it to study quantum CNOT gates with coherent states.
We conclude in section IX with a discussion of our results.

\section{Preliminaries}

\subsection{Coherent states}
Let $h$ be the harmonic oscillator Hilbert space. We denote with ${\cal O}$,
its zero-dimensional subspace that contains only the zero vector (which does not represent a quantum state).
We also denote with lower case $h_i$ general subspaces of $h$, and with upper case $H_i$ its coherent subspaces, which are introduced later.
$a,a^{\dagger}$ are the annihilation and creation operators, and 
\begin{eqnarray}
x=2^{-1/2}(a+a^{\dagger});\;\;\;\;\;p=2^{-1/2}i(a^{\dagger}-a)
\end{eqnarray}
are the position and momentum operators. 
$D(A)$ are the displacement operators
\begin{eqnarray}
D(A)=\exp(A a^{\dagger}-A^*a);\;\;\;\;\;A\in {\mathbb C}.
\end{eqnarray}
Coherent states\cite{coh1,coh2,coh3} are defined as
\begin{eqnarray}
\ket{A}=D(A)\ket{0}=\exp\left (-\frac{|A|^2}{2}\right )\sum _{N=0}^\infty \frac{A^N}{\sqrt{N!}}\ket {N};\;\;\;\;\;a\ket{A}=A\ket{A}
\end{eqnarray}
where $\ket{N}$ are number eigenstates.

Let $H(A)$ be the one-dimensional space that contains the coherent states $\ket{A}$, and $\Pi(A)$ the corresponding projector:
\begin{eqnarray}
\Pi (A)=\ket{A}\bra{A}.
\end{eqnarray}
We will use the notation
\begin{eqnarray}
\Pi^{\perp}(A)={\bf 1}-\Pi (A).
\end{eqnarray}
The $\Pi ^{\perp}(A)$ are not trace class operators.
It is easily seen that
\begin{eqnarray}\label{nnn}
D(z)\Pi (A)[D(z)]^{\dagger}=\Pi (A+z);\;\;\;\;\;D(z)\Pi ^{\perp}(A)[D(z)]^{\dagger}=\Pi ^{\perp}(A+z).
\end{eqnarray}

An important property of coherent states is the resolution of the identity
\begin{eqnarray}
\int _{\mathbb C}\frac{d^2A}{\pi}\Pi(A)={\bf 1}.
\end{eqnarray}
Another property (closure under time evolution) is that
under time evolution with the Hamiltonian $a^{\dagger}a$, a coherent state evolves into other coherent states:
\begin{eqnarray}\label{1111}
\exp (ita^{\dagger}a) \Pi (A) \exp (-ita^{\dagger}a)=\Pi [A\exp (it)].
\end{eqnarray}

\subsection{The Bargmann representation}
We consider the quantum state
\begin{eqnarray}\label{gen}
\ket{s}=\sum _{N=0}^{\infty}s_N\ket{N};\;\;\;\;\;\sum _{N=0}^{\infty}|s_N|^2=1
\end{eqnarray}
and use the notation
\begin{eqnarray}\label{gen1}
\bra{s}=\sum _{N=0}^{\infty}s_N^*\bra{N};\;\;\;\;\;\bra{s^*}=\sum _{N=0}^{\infty}s_N\bra{N};\;\;\;\;\;\ket{s^*}=\sum _{N=0}^{\infty}s_N^*\ket{N}.
\end{eqnarray}
In the Bargmann representation\cite{A1,A2,vour}, the state $\ket{s}$ is represented with the analytic function
\begin{eqnarray}\label{6mz}
\ket{s}\;\rightarrow\; s(z)=\exp \left (\frac{1}{2}|z|^2\right )\langle z^*\ket{s}=\sum _{N=0}^{\infty}\frac{s_Nz^N}{\sqrt{N!}}.
\end{eqnarray}
The scalar product is given by
\begin{eqnarray}\label{10}
\langle g\ket{s}=\int \frac{d^2z}{\pi}[g(z)]^*s(z)\exp(-|z|^2).
\end{eqnarray}

As an example, we consider the coherent state $\ket{A}$ for which the Bargmann function is
\begin{eqnarray}\label{9800}
\ket{A}\;\rightarrow \;f(z)=\exp\left (Az-\frac{1}{2}|A|^2\right ).
\end{eqnarray} 

 We next review briefly results about the growth\cite{C1,C2,C3} of Bargmann functions, and the density of their zeros.
They are needed in section \ref{totalset}.
\begin{definition}
The growth of an analytic function $s(z)$ is described with the order $\rho$ and type $\sigma$, given by
\begin{eqnarray}\label{dv1}
\rho=\limsup  _{R\rightarrow \infty}\frac{\ln \ln M(R)}{\ln R};\;\;\;\;\sigma= \limsup  _{R\rightarrow \infty}\frac{\ln M(R)}{R^\rho}.
\end{eqnarray}
Here $M(R)$ is the maximum value of $|s(z)|$ on the circle $|z|=R$.
A function $s(z)$ with growth $(\rho, \sigma)$ grows at large distances $R$ from the center, as $|s(z)|\approx \exp (\sigma R^\rho )$.
\end{definition}
A growth $(\rho_1, \sigma_1)$ is smaller than the growth $(\rho _2, \sigma _2)$, if $\rho _1<\rho _2$ (in which case we do not compare 
$\sigma_1,\sigma_2$) or if $\rho _1=\rho _2$ and $\sigma _1<\sigma _2$ (lexicographic order). We denote this as  $(\rho _1, \sigma _1)\prec (\rho _2, \sigma _2)$.

As an example, we consider the coherent state $\ket{A}$ for which the Bargmann function has growth $(1,|A|)$.

\begin{definition}
Consider a sequence of complex numbers $A _1,A_2,...$ such that
\begin{eqnarray}\label{vbn}
0<|A _1|\le |A _2|\le .... ;\;\;\;\;\;\lim _{N\rightarrow \infty}|A_N|=\infty
\end{eqnarray}
Let $n(R)$ be the number of terms of this sequence within the circle $|A|<R$.
The density of this sequence is described with the numbers
\begin{eqnarray}\label{102}
\eta=\limsup _{R\to \infty }\frac{\ln n(R)}{\ln R};\;\;\;\;\;\delta =\lim _{R\to \infty }\frac {n(R)}{R^\eta}
\end{eqnarray}
The number of terms of this sequence in a large circle with radius $R$ is $n(R)\approx \delta R^\eta$.
\end{definition}
We say that the density $(\eta, \delta)$ of a sequence is greater than $(\eta_1, \delta_1)$ if
$\eta > \eta_1$ and also if $\eta=\eta _1$ and $\delta>\delta_1$ (lexicographic order).

There are theorems that relate the growth of analytic functions with the density of their zeros.
Using these ideas in the context of Bargmann functions, we arrive at the following proposition, which we present without proof
(e.g., \cite{VO1,VO2} and references therein)
\begin{proposition}\label{580}
\mbox{}
\begin{itemize}
\item[(1)]
The growth of Bargmann functions is less than $(2,\frac{1}{2})$.
\item[(2)]
 The density of zeros of Bargmann functions is smaller than $(2,1)$.
 \end{itemize}
\end{proposition}

\section{Dirac's contour representation}

In the Dirac contour representation \cite{DI1,DI2,DI3,DI4,DI5} the ket states are represented by a different function than the bra states.
The ket state $\ket{s}$ in Eq.(\ref{gen}) and 
the corresponding bra state $\bra{s}$ in Eq.(\ref{gen1}), are represented by the functions:
\begin{eqnarray}\label{17}
\ket{s}\;\rightarrow\; s_k(z)=\sum _{N=0}^{\infty}\frac{s_Nz^N}{\sqrt{N!}}
;\;\;\;\;\bra{s}\;\rightarrow\;s_b(z)=\sum_{N=0}^\infty \frac{s_N^*\sqrt{N!}}{z^{N+1}}.
\end{eqnarray}
The indices $k,b$ in the notation indicate `ket functions' and `bra functions', correspondingly.
It is seen that ket states are represented in exactly the same way as in the Bargmann representation.
$s_k(z)$ is an analytic function, while $s_b(z)$ has singularities which play a crucial role in the formalism.

The scalar product is given by
\begin{eqnarray}
\langle f\ket{s}=\oint _C \frac{dz}{2\pi i}f_b(z) s_k(z)=\sum _{N=0}^{\infty}f_N^*s_N.
\end{eqnarray}
Here $C$ is a simple anticlockwise contour enclosing the singularities of $f_b(z)$.

The functions $s_k(z)$ and $s_b(z)$ are related through the transforms
\begin{eqnarray}
\oint _C \frac{dz}{2\pi i}s_b(z) \exp(\zeta ^*z)=[s_k(\zeta )]^*;\;\;\;\;\;
s_b(z)=\frac{1}{z}\int _0^{\infty} dt \exp (-t)\left[s_k\left (\frac{t}{z^*}\right )\right ]^*.
\end{eqnarray}

As an example, we consider the coherent state $\ket{A}$. In this case 
\begin{eqnarray}\label{98}
&&\ket{A}\;\rightarrow \;s_k(z)=\exp\left (Az-\frac{1}{2}|A|^2\right )\nonumber\\
&&\bra{A}\;\rightarrow \;s_b(z)=\frac{\exp(-\frac{1}{2}|A|^2)}{z-A^*};\;\;\;|z|>|A|
\end{eqnarray} 
Here the sum for $s_b(z)$ converges to the above result, only for $|z|>|A|$.
We define the annulus of a contour $C$, to be the region $r_{\rm min}\le |z|\le r_{\rm max}$ where $r_{\rm min}$ ($r_{\rm max}$) is the 
minimum (maximum) value of $|z|$ on the contour.
Then convergence requires that the annulus of $C$ should enclose the pole. 
More generally, the convergence requirements in the formalism are restrictions on the contour.

An operator 
\begin{eqnarray}
\Theta =\sum \Theta _{MN}\ket{M}\bra{N}
\end{eqnarray}
is represented by the function
\begin{eqnarray}\label{33}
\Theta (z_1,z_2)=\sum \Theta _{MN}\sqrt {\frac {N!}{M!}}\frac{z_1^M}{z_2^{N+1}},
\end{eqnarray}
and it acts on ket states as
\begin{eqnarray}\label{1959}
\Theta \ket{s}\;\rightarrow\; \oint _C \frac{d\zeta }{2\pi i}\Theta (z,\zeta)s_k(\zeta)=\sum  \Theta _{MN}s_N\ket{M},
\end{eqnarray}
and on bra states as
\begin{eqnarray}\label{194}
\bra {s}\Theta \;\rightarrow\; \oint _C \frac{d\zeta}{2\pi i}s_b(\zeta)\Theta (\zeta,z)=\sum  \Theta _{MN}s_M^*\ket{N}.
\end{eqnarray}
Its trace is given by
\begin{eqnarray}\label{tr}
{\rm Tr}\Theta =\oint _{C_1} \frac{dz_1}{2\pi i}\oint _{C_2} \frac{dz_2}{2\pi i}\frac{\Theta (z_1,z_2)}{z_2-z_1}=\sum  \Theta _{NN};\;\;\;\;|z_1|>|z_2|.
\end{eqnarray}
Here convergence requires that $|z_1|>|z_2|$, which means that the annulus of $C_2$ 
should be entirely inside the annulus of $C_1$. 
 
Products of two operators are given by
\begin{eqnarray}\label{product}
\Theta _1\Theta _2\;\rightarrow\; \oint _C \frac{d\zeta }{2\pi i}\Theta _1(z_1,\zeta)\Theta _2(\zeta , z_2).
\end{eqnarray}
If $\Theta$ is a projector then $\Theta ^2=\Theta$ and 
\begin{eqnarray}\label{prod1}
\oint _C \frac{d\zeta }{2\pi i}\Theta (z_1,\zeta)\Theta (\zeta , z_2)=\Theta (z_1,z_2).
\end{eqnarray}

As examples we consider the operators:
\begin{eqnarray}\label{26}
&&{\bf 1}\;\rightarrow\;\Theta (z_1,z_2)=\frac{1}{z_2-z_1};\;\;\;\;|z_2|>|z_1|\nonumber\\
&&\Pi(A)\;\rightarrow\;\Theta (z_1,z_2)=\frac{\exp\left (Az_1-|A|^2\right )}{z_2-A^*};\;\;\;\;|z_2|>|A|\nonumber\\
&&\ket{A_1}\bra{A_2}\;\rightarrow\;\Theta (z_1,z_2)=\frac{\exp\left (A_1z_1-\frac{1}{2}|A_1|^2-\frac{1}{2}|A_2|^2\right )}{z_2-A_2^*};\;\;\;\;|z_2|>|A_2|.
\end{eqnarray}
There are many technical details related to this formalism, some of which are discussed in \cite{DI4,DI5}. 
The region of convergence in the sum of Eq.(\ref{17}) for the bra functions,  needs to be studied on an individual basis.
In this paper we are interested in finite superpositions of coherent states, in which case the bra functions have a finite number of poles.

If $f_b(z)$ and $g_b(z)$ are two bra functions with finite sets of poles $S_1$ and $S_2$,
then the set of poles of the superposition $\lambda f_b(z)+\mu g_b(z)$ is $S_1\cup S_2$, or a subset of $S_1\cup S_2$ if there are cancellations
(e.g., if both $f_b(z)$ and $g_b(z)$ have the same term $\alpha /(z-A^*)$, then $f_b(z)-g_b(z)$ does not have a pole at $A^*$).
Analogous result holds for a superposition of a finite number of bra functions, but it does not hold for a superposition of an
infinite number of bra functions.
For example, the number state $\ket {N}$ is represented by the bra function $1/z^{N+1}$ and has a pole at zero.
Any finite superposition of number states has a pole at zero,
but the coherent state $\ket{A}$, which is an infinite superposition of number states, has a pole at $A^*$.

\section{Coherent subspaces}
\subsection{The Birkhoff-von Neumann lattice $\Lambda (h)$}

We consider the set $\Lambda (h)$ of closed subspaces of $h$.
For any two elements $h_1,h_2$ we define the 
disjunction (logical `OR\rq{}) operation $h_1\vee h_2$, which is the closed subspace that contains all superpositions of vectors in the subspaces $h_1,h_2$:
\begin{eqnarray}\label{1}
h_1 \vee h_2={\overline {\rm span}}[h_1 \cup h_2].
\end{eqnarray}
The overline indicates closure.
The quantum `OR' is different from the classical `OR' (which is the union in set theory).
The quantum `OR' is more than the union of two spaces, because it involves all superpositions.
The coherent spaces and coherent projectors below, are based on the quantum `OR'.

The conjuction (logical `AND') operation is
\begin{eqnarray}\label{2}
h_1 \wedge h_2=h_1 \cap h_2.
\end{eqnarray}
The relevant partial order is `subspace', and we use the notation $h_1 \prec h_2$ to indicate that $h_1$ is a subspace of $h_2$.
The negation (logical `NOT\rq{}) of $h_1$ is its orthocomplement $h_1^{\perp}$ (i.e., an orthogonal space to $h_1$ such that $h_1 \vee h_1^{\perp}=h$).
The set $\Lambda (h)$ with these operations is the Birkhoff-von Neumann lattice of closed subspaces of $h$,
that describes the logic of quantum mechanics\cite{LO1,LO2,LO3}.

\subsection{Coherent subspaces}

We use coherent states to introduce `coherent subspaces\rq{} of $h$, and a `coherent lattice' ${\cal L} _{\rm coh}$.
We consider the following two-dimensional space $H(A_1,A_2)=H(A_1) \vee H(A_2)$, that contains the superpositions of 
the coherent states $\alpha _1\ket{A_1}+\alpha _2 \ket{A_2}$.
In the Bargmann language it contains the functions 
\begin{eqnarray}
f(z)=\lambda _1 \exp (A_1z)+\lambda _2\exp(A_2z).
\end{eqnarray}
Inductively, we generalize this to a finite number of coherent states.
\begin{proposition}\label{2wd}
Any finite number of coherent states, are linearly independent.
\end{proposition}
\begin{proof}
We consider the coherent states $\ket {A_1},...,\ket{A_i}$ and show that
\begin{eqnarray}\label{98a}
\lambda _1\ket{A_1}+...+\lambda _i\ket{A_i} =0\;\;\rightarrow\;\;\lambda _1=...=\lambda _i=0.
\end{eqnarray} 
This can written as
\begin{eqnarray}\label{98b}
\sum _{N=0}^\infty \mu _N\ket{N}=0;\;\;\;\;\mu _N=\sum _{j=1}^i\lambda _j\exp\left (-\frac{1}{2}|A_j|^2\right )\frac{A_j^N}{\sqrt{N!}}.
\end{eqnarray}  
From this follows that $\mu _N=0$, for all $N$. 
We need to satisfy an infinite number of equations with a finite number of unknowns,
and the only solution is $\lambda _1=...=\lambda _i=0$.
\end{proof}
\begin{definition}\label{defi}
Let $S=\{A_1,...,A_i\}$ be a finite set of complex numbers, and $S^*=\{A_1^*,...,A_n^*\}$. 
\begin{itemize}
\item[(1)]
The coherent subspace of $h$ specified by $S$, is
\begin{eqnarray}
H(S)=H(A_1,...,A_i)=H(A_1)\vee...\vee H(A_i),
\end{eqnarray}
and contains all the superpositions $\alpha _1\ket{A_1}+...+\alpha _i\ket{A_i}$.
The space $H(S)$ is $i$-dimensional, and the coherent states $\ket {A_1},...,\ket{A_i}$ form a non-orthogonal basis in it.
For $S=\emptyset$, the $H(\emptyset)={\cal O}$ (this is different from the
one-dimensional space $H(0)$ that contains the vacuum).
We denote as $\Pi(S)=\Pi(A_1,...,A_i)$ the projector to the space $H(S)=H(A_1,...,A_i)$, and as
\begin{eqnarray}
\Pi^{\perp}(A_1,...,A_i)={\bf 1}-\Pi(A_1,...,A_i).
\end{eqnarray}
\item[(2)]
In the Bargmann language, the coherent subspace specified by $S$, contains the 
`coherent Bargmann functions\rq{} 
\begin{eqnarray}\label{4vv}
f(z)=\lambda _1 \exp (A_1z)+...+\lambda _i\exp(A_iz).
\end{eqnarray}
The growth of these functions is less or equal to $\rho=1$ and $\sigma =\max(|A_1|,...|A_i|)$.
\item[(3)]
In Dirac's contour representation language, 
the coherent subspace specified by $S$, contains the `coherent ket\rq{} functions of Eq.(\ref{4vv}), and the `coherent bra\rq{} functions 
\begin{eqnarray}\label{4v}
f_b(z)=\frac{\lambda _1}{z-A_1^*}+...+\frac{\lambda _i}{z-A_i^*}
\end{eqnarray}
The set of poles of the coherent bra functions is $S^*$ or a subset of $S^*$ (if some of the $\lambda _k$ are zero).
\end{itemize}
\end{definition}
We consider two states in $H(S)$:
\begin{eqnarray}\label{44}
\ket{u}=\sum _{j=1}^i u _j\ket{A_j};\;\;\;\;\ket{v}=\sum _{j=1}^i v _j\ket{A_j}.
\end{eqnarray}
Their overlap is given by
\begin{eqnarray}\label{prod}
\langle {u}\ket{v}=\sum _{j,k} u _j^*v _k g_{jk}(S)
\end{eqnarray}
where $g(S)$ is the $i\times i$ Hermitian matrix of rank $i$, with elements
\begin{eqnarray}\label{matr}
g_{jk}(S)=\langle A_j\ket{A_k}=\exp \left (A_j^*A_k-\frac{1}{2}|A_j|^2-\frac{1}{2}|A_k|^2\right ).
\end{eqnarray}
Its diagonal elements are $g_{jj}(S)=1$.
The normalization condition for the vector $\ket {u}$ is
\begin{eqnarray}\label{norma}
\sum _{j,k} u _j^*u _k g_{jk}(S)=1.
\end{eqnarray}
It is seen that $g(S)$ is a positive-definite matrix.
Let $\gamma _j(S)$, $e_j(S)$, $E_j(S)$ be its eigenvalues, eigenvectors and eigenprojectors:
\begin{eqnarray}\label{eigv}
g(S)e_j(S)=\gamma _j(S)e_j(S);\;\;\;E_j(S)=e_j(S)[e_j(S)]^{\dagger};\;\;\;\gamma_j(S)>0.
\end{eqnarray}
Then
\begin{eqnarray}
g(S)=\sum _j\gamma _j(S)E_j(S);\;\;\;E_j(S)E_k(S)=E_j(S)\delta_{jk};\;\;\;\sum _jE_j(S)=\Pi(S).
\end{eqnarray}
For later use we also define its inverse matrix $G(S)$:
\begin{eqnarray}\label{pd}
G(S)=[g(S)]^{-1}=\sum _j\frac{1}{\gamma _j(S)}E_j(S).
\end{eqnarray}
The inverse of the matrix $g$ exists because a finite number of coherent states are linearly independent (both matrices $g$ and $G$ are of rank $i$).

We note that in the limit $\min _{k,j}(|A_j-A_k|)\rightarrow \infty$ the off-diagonal elements of the matrix $g$ become zero and 
$g\rightarrow {\bf 1}$.

The set of poles of the bra functions in Dirac's contour representation, links the algebraic structure
of finite sets of complex numbers (studied in section \ref{X2}),
to the coherent subspaces.
The term coherent subspaces refers to the properties of the corresponding projectors $\Pi(S)=\Pi(A_1,...,A_i)$, which are discussed in section \ref{coh}. 

\begin{remark}
An arbitrary finite-dimensional subspace of $h$, is not in general a coherent subspace. 
For example, the one dimensional subspace that contains the vector ${\cal K}[\ket{A_1}+\ket{A_2}]$
(where ${\cal K}$ is a normalization constant), is not a  coherent subspace.
There is no coherent state $\ket{A_3}$ equal to ${\cal K}[\ket{A_1}+\ket{A_2}]$ (because the $\ket{A_1}$, $\ket{A_2}$, $\ket{A_3}$ are linearly independent). Only some finite-dimensional subspaces of $h$ are coherent subspaces, and they are labeled with finite sets of complex numbers.
\end{remark}
For later use we give the following example.
\begin{example}\label{ex1}
In the case $S=\{A_1,A_2\}$, the corresponding $g$-matrix is
\begin{eqnarray}\label{bbnn}
&&g(A_1,A_2)=
\begin{pmatrix}
1&\mu e^{i\theta}\\
\mu e^{-i\theta}&1
\end{pmatrix};\;\;\;\;
\mu =\exp \left (-\frac{1}{2}|A_1-A_2|^2\right )\nonumber\\
&& \theta =\frac{i}{2} (A_2^*A_1-A_1^*A_2)
\end{eqnarray}
The inverse of $g(A_1,A_2)$ is:
\begin{eqnarray}\label{inverse}
&&G(A_1,A_2)=[g(A_1,A_2)]^{-1}=\frac{1}{1-\mu ^2}
\begin{pmatrix}
1&-\mu e^{i\theta}\\
-\mu e^{-i\theta}&1
\end{pmatrix}.
\end{eqnarray}
The eigenvalues of the $g$ matrix are $\gamma_1=1+\mu$ and $\gamma _2=1-\mu$ and the corresponding eigenvectors are
\begin{eqnarray}\label{pd15}
&&e_1(A_1,A_2)=\frac{1}{\sqrt 2}
\begin{pmatrix}
e^{i\theta}\\
1
\end{pmatrix};\;\;\;\;
e_2(A_1,A_2)=\frac{1}{\sqrt 2}
\begin{pmatrix}
-e^{i\theta}\\
1
\end{pmatrix}
\end{eqnarray}
We also introduce the corresponding projectors
\begin{eqnarray}
E_1(A_1,A_2)=\frac{1}{2}
\begin{pmatrix}
1 &e^{i\theta}\\
e^{-i\theta} &1
\end{pmatrix};\;\;\;\;
E_2(A_1,A_2)=\frac{1}{2}
\begin{pmatrix}
1 &-e^{i\theta}\\
-e^{-i\theta} &1
\end{pmatrix}
\end{eqnarray}

\end{example}

\subsection{Coherent projectors }

We consider the projector $\Pi(A_1,A_2)$ to the space $H(A_1,A_2)$ and we prove that for $A_1\ne A_2$
\begin{eqnarray}\label{35}
\Pi(A_1,A_2)&=&{\tau}_2\left[\Pi(A_1)+\Pi_1(A_2)-\Pi(A_1)\Pi(A_2)-\Pi(A_2)\Pi(A_1)\right ]
\nonumber\\&=&{\tau}_2\left[\Pi ^{\perp}(A_2)\Pi (A_1)+\Pi ^{\perp}(A_1)\Pi(A_2)\right ],
\end{eqnarray}
where
\begin{eqnarray}
&&\tau_2=\{{\rm Tr}[\Pi^{\perp}(A_1)\Pi(A_2)]\}^{-1}=\frac{1}{1-\mu ^2}.
\end{eqnarray}
$\mu$ has been given in Eq.(\ref{bbnn}).
We prove this using the Gram-Schmidt orthogonalization algorithm.
We take the component of $\ket{A_2}$ which is perpendicular to $\ket {A_1}$, normalize it into the vector $\ket{u_2}$ with length equal to one,
and then add $\Pi(A_1)$ and $\ket{u _2}\;\bra{u _2}$:
\begin{eqnarray}\label{alg}
&&\Pi (A_1,A_2)=\Pi (A_1)+\ket{u_2}\;\bra{u_2};\;\;\;\;\ket{u_2}=\sqrt {{\tau}_2}\Pi ^{\perp}(A_1)\ket{A_2}.
\end{eqnarray}
This can be written in a compact way, as
\begin{eqnarray}\label{alg}
&&\Pi(A_1,A_2)=\varpi(A_1)+\varpi (A_2|A_1)\nonumber\\
&&\varpi (A_1)=\Pi (A_1);\;\;\;\;\varpi (A_2|A_1)=\ket{u_2}\;\bra{u_2}=\frac{\Pi ^{\perp}(A_1)\Pi (A_2)\Pi ^{\perp}(A_1)}{{\rm Tr}[\Pi ^{\perp}(A_1)\Pi (A_2)]}
\end{eqnarray}
The $\varpi (A_2|A_1)$ is a projector (to a one-dimensional space) orthogonal to $\Pi (A_1)$, and it is equal to $\Pi (A_1,A_2)-\Pi (A_1)$. 
We insert $\Pi^{\perp}(A)={\bf 1}-\Pi(A)$ into Eq.(\ref{alg}), and we get Eq.(\ref{35}).

The $\Pi(A_1,A_2)$ can also be written in terms of the matrix $G(A_1,A_2)$ in Eq.(\ref{inverse}), as
\begin{eqnarray}
\Pi(A_1,A_2)=\sum _{j,k}G_{jk}(A_1,A_2)\ket{A_j}\bra{A_k};\;\;\;\;j,k=1,2.
\end{eqnarray}

If $\rho$ is a density matrix, the $Q(A)={\rm Tr}[\rho\Pi(A)]$ is the probability that a measurement with the projector $\Pi(A)$
will give `yes'. But since the $\Pi(A)$ are not orthogonal projectors to each other, the $Q(A)$ as a distribution (known as the 
Husimi $Q$-function), is not a real probability distribution. It is a quasi-probability distribution 
of a quantum particle being at the point $A$ in phase space, and it has the property
\begin{eqnarray}\label{3sd}
\int \frac {d^2A}{\pi}Q(A)=1.
\end{eqnarray}
In a similar way the $Q(A_1,A_2)={\rm Tr}[\rho\Pi(A_1,A_2)]$ is the probability that a measurement with the projector $\Pi(A_1,A_2)$
will give `yes'. Here also the $\Pi(A_1,A_2)$ are not orthogonal projectors to each other, and the $Q(A_1,A_2)$ (which is a generalization of the 
Husimi $Q$-function), is not a real probability distribution. It is a quasi-probability distribution 
of a quantum particle being at the point $A_1$ `OR' $A_2$ in phase space.
This is the quantum `OR' in the Birkhoff-von Neumann lattice that involves superpositions, and it is different
from the classical `OR' in Boolean algebras.
The space $H(A_1,A_2)$ contains superpositions of coherent states (`Schr\"odinger cats')
and the discussion on Schr\"odinger cats elucidates the difference between the quantum and classical `OR'.
In Eq.(\ref{9850}) below it is shown that the $Q(A_1,A_2)$ obeys a relation analogous to Eq.(\ref{3sd}).

Inductively, we generalize the above formalism to a finite number of coherent states.
Using the Gram-Schmidt orthogonalization algorithm, we express the projector $\Pi (A_1,...,A_i)$ to the space $H(A_1,...,A_i)$, as
\begin{eqnarray}\label{alg10}
&&\Pi(A_1,...,A_i)=\Pi(A_1,...,A_{i-1})+\ket{u_i}\;\bra{u_i};\;\;\;A_j\ne A_k\nonumber\\
&&\ket{u_i}=\sqrt {\tau _i}\Pi ^{\perp}(A_1,...,A_{i-1})\ket{A_i}\nonumber\\
&&\tau_i=\{{\rm Tr}[\Pi^{\perp}(A_1,...,A_{i-1})\Pi(A_i)]\}^{-1}.
\end{eqnarray}
The linear independence of the coherent states that we proved above, ensures that the ${\rm Tr}[\Pi^{\perp}(A_1,...,A_{i-1})\Pi(A_i)]$ is different from zero. 
We rewrite this as
\begin{eqnarray}\label{alg1}
&&\Pi(A_1,...,A_i)=\Pi(A_1,...,A_{i-1})+\varpi (A_i |A_1,...,A_{i-1})\nonumber\\
&&\varpi (A_i |A_1,...,A_{i-1})=\ket{u_{i}}\bra{u_{i}}=\frac{\Pi  ^{\perp} (A_1,...,A_{i-1}) \Pi (A_i)\Pi ^{\perp}(A_1,...,A_{i-1})}
{{\rm Tr}[\Pi ^{\perp}(A_1,...,A_{i-1})\Pi (A_i)]};\;\;\;\;A_j\ne A_k\nonumber\\
&&{\rm Tr}[\Pi (A_1,...,A_i)]=i;\;\;\;\;\;{\rm Tr}[\varpi (A_i |A_1,...,A_{i-1})]=1.
\end{eqnarray}
We call the $\Pi (A_1,...,A_i)$, $\varpi (A_i |A_1,...,A_{i-1})$ coherent projectors, because 
they have properties analogous to coherent states, as discussed in section \ref{coh}.

The $\varpi (A_1), \varpi (A_2|A_1),...,\varpi (A_n|A_1,...,A_{n-1})$ are projectors orthogonal to each other, and
$\varpi (A_n |A_1,...,A_{n-1})\ne \Pi(A_n)$. From Eq.(\ref{alg1}), it follows that
\begin{eqnarray}\label{59}
&&\Pi (A_1,...,A_n)=\varpi (A_1)+\varpi (A_2|A_1)+...+\varpi (A_n|A_1,...,A_{n-1}).
\end{eqnarray}
Let
\begin{eqnarray}\label{sets}
S_1=\{A_1,...,A_j\};\;\;\;\;S_2=\{B_1,...,B_i\};\;\;\;S_1\cap S_2=\emptyset.
\end{eqnarray}
From Eq.(\ref{59}), it follows that
\begin{eqnarray}
&&\Pi(S_1\cup S_2)=\Pi(S_1)+\varpi(S_2|S_1);\;\;\;\Pi(S_1)\varpi(S_2|S_1)=0\nonumber\\
&&\varpi(S_2|S_1)=\varpi (B_1|S_1)+\varpi (B_2|S_1\cup\{B_1\})+...+\varpi (B_i|S_1\cup \{B_1,...,B_{i-1}\}).
\end{eqnarray}

The $Q(A_1,...,A_i)={\rm Tr}[\rho\Pi(A_1,...,A_i)]$ is the probability that a measurement with the projector $\Pi(A_1,...,A_i)$
will give `yes'. The generalized
Husimi $Q$-function $Q(A_1,...,A_i)$ is a quasi-probability distribution 
of a quantum particle being at the point $A_1$ `OR' $A_2$ `OR' $A_3$, etc (the quantum `OR' that involves superpositions).
An integral analogous to Eq.(\ref{3sd}) is given in Eq.(\ref{9850}) below.

\subsection{Total sets of coherent subspaces}\label{totalset}

\begin{definition}
A set of subspaces $\{h_i\}$ is called total, if there is no state in $h$, which is orthogonal to all $h_i$. 
\end{definition}
\begin{lemma}\label{LL1}
A state $\ket{s}$ is orthogonal to the coherent subspace $H(A_1,...,A_i)$ (i.e., $\Pi(A_1,...,A_i)\ket{s}=0$), if and only if the $A_1^*,...,A_i^*$ are zeros of its Bargmann function $s(z)$ (i.e., $s(A_j^*)=0$ for $j=1,...,i$).
\end{lemma}
\begin{proof}
From Eq.(\ref{6mz}), it follows that 
a state $\ket{s}$ is orthogonal to the coherent state $\ket{\zeta }$, if and only if $\zeta ^*$ is a zero of its Bargmann function $s(z)$. 
If  $\ket{s}$ is orthogonal to the coherent subspace $H(A_1,...,A_i)$, then it is orthogonal to the coherent states $\ket{A_j}$ and therefore it has the 
$A_j^*$ as zeros.
\end{proof}

\begin{proposition}\label{total}
\mbox{}
\begin{enumerate}
\item[(1)]
A set of coherent subspaces which is uncountably infinite, is a total set.
\item[(2)]
Let $\{H(S_i)\}$ be a countably infinite set of coherent subspaces with $S_i=\{A_{i1},...,A_{i k_i}\}$.
We use a lexicographic order and relabel the $A_{ij}$ as $A_n$.
\begin{itemize}
\item
If the sequence $A_n$ converges to some point $A$, the $\{H(S_i)\}$ is a total set of coherent subspaces.
\item
If the sequence $|A_n|$ diverges,  and it has density greater than $(2,1)$ then the $\{H(S_i)\}$ is a total set of coherent subspaces.
\end{itemize}
\end{enumerate}
\end{proposition}
\begin{proof}
\mbox{}
\begin{itemize}
\item[(1)]
The zeros of analytic functions are isolated. 
Therefore the total number of zeros of an analytic function is at most countably infinite. According to lemma \ref{LL1}, 
a state which is orthogonal to all subspaces in an uncountably infinite set of coherent subspaces,
would have an uncountably infinite set of zeros.
Consequently, a set of coherent subspaces which is uncountably infinite, is a total set.
\item[(2)]
According to lemma \ref{LL1}, a state which is orthogonal to all subspaces $\{H(S_i)\}$, will have the $A_N^*$ as zeros.
The zeros of analytic functions are isolated, and they can not converge to some point $A$.
Also in the case that the $|A_n|$ diverges, proposition \ref{580} shows that its density should be smaller than $(2,1)$.
Therefore in both cases, the $\{H(S_i)\}$ is a total set of coherent subspaces.
\end{itemize}
\end{proof}
For some total sets of coherent subspaces, we might be able to find a resolution of the identity in terms of the corresponding projectors. An example of this is given below.

\section{Properties of coherent projectors}\label{coh}

In this section we show that the coherent projectors $\Pi(A_1,...,A_n)$ have properties similar to those of $\Pi (A)$, i.e., to coherent states.

The following proposition is a resolution of the identity.
\begin{proposition}\label{propo}
\mbox{}
\begin{itemize}
\item[(1)]
The following resolution of the identity holds in terms of the  projectors ${\Pi}(A,A+d_2,...,A+d_n)$
(of rank $n$), with fixed $d_2,...,d_n$:
\begin{eqnarray}\label{985A}
\int _{\mathbb C}\frac{d^2A}{n\pi}{\Pi}(A,A+d_2,...,A+d_n)={\bf 1}
\end{eqnarray}
\item[(2)]
The following resolution of the identity holds in terms of the  projectors 
${\varpi}(A|A+d_1,...,A+d_{n-1})$ (of rank $1$), with fixed $d_1,...,d_{n-1}$:
\begin{eqnarray}\label{99}
\int _{\mathbb C}\frac{d^2A}{\pi}{\varpi} (A|A+d_1,...,A+d_{n-1})={\bf 1}
\end{eqnarray}

\end{itemize}
\end{proposition}
\begin{proof}
\mbox{}
\begin{itemize}
\item[(1)]
According to Eq.(38) in ref.\cite{vour},
for any trace class operator $\Theta$:
\begin{eqnarray}\label{bb}
\int _{\mathbb C}\frac{d^2z}{\pi}D(z)\Theta [D(z)]^{\dagger}={\bf 1}{\rm Tr}\Theta.
\end{eqnarray}
We use this relation with $\Theta={ \Pi}(0,d_2,...,d_n)$ (in which case ${\rm Tr}\Theta=n$),
and we get Eq.(\ref{985A}).
\item[(2)]
We use Eq.(\ref{bb}) with ${\varpi}(A|A+d_1,...,A+d_{n-1})$ (in which case ${\rm Tr}\Theta=1$), and we get Eq.(\ref{99}).
\end{itemize}
\end{proof}
The following proposition proves a closure property, under displacement transformations and time evolution.
\begin{proposition}\label{propo1}
\mbox{}
\begin{itemize}

\item[(1)]
Under displacement transformations the $\Pi (A_1,...,A_n)$, $\varpi (A_n |A_1,...,A_{n-1})$, are transformed into projectors of the same type:
\begin{eqnarray}\label{1121}
&&D(z) \Pi (A_1,...,A_n) [D(z)]^{\dagger}=\Pi (A_1+z,...,A_n+z)\nonumber\\
&&D(z) \varpi  (A_n |A_1,...,A_{n-1}) [D(z)]^{\dagger}=\varpi (A_n+z |A+z,...,A_{n-1}+z).
\end{eqnarray}
This is analogous to Eq.(\ref{nnn}) for coherent states.
\item[(2)]
Under time evolution with the Hamiltonian $a ^{\dagger}a$, the $\Pi (A_1,...,A_n)$, $\varpi (A_n |A_1,...,A_{n-1})$, are transformed into projectors of the same type:
\begin{eqnarray}\label{1121A}
&&\exp (ita ^{\dagger}a) \Pi (A_1,...,A_n) \exp (-ita ^{\dagger}a )=\Pi [A_1\exp(it),...,A_n\exp(it)]\nonumber\\
&&\exp (ita^{\dagger}a) \varpi (A_n |A_1,...,A_{n-1}) \exp (-ita ^{\dagger}a )=\varpi [A_n \exp(it)|A_1\exp(it),...,A_{n-1}\exp(it)].
\end{eqnarray}
This is analogous to Eq.(\ref{1111}) for coherent states.
\end{itemize}
\end{proposition}
\begin{proof}
\mbox{}
\begin{itemize}

\item[(1)]
For  $\varpi (A_n |A_1,...,A_{n-1})$, we prove the statement inductively using Eq.(\ref{alg1}) and the property of coherent states in Eq.(\ref{nnn}).
Then Eq.(\ref{59}) proves the statement for  $\Pi (A_1,...,A_n)$.
\item[(2)]
The proof here is similar to the one for the first case.
\end{itemize}
\end{proof} 

Coherent states are eigenstates of $a$, which we can express as $a^\ell\Pi(A_1)=A_1^{\ell}\Pi(A_1)$.
Weaker statements than this that involve the $\varpi(A_{n+1}|A_1,...,A_{n})a^{\ell}\varpi(A_{n+1}|A_1,...,A_{n})$ and also the 
trace of $a^{\ell}\Pi (A_1,...,A_{n})$ are made below.
\begin{proposition}\label{bnm}
\begin{itemize}
For $\ell=1,2,...$
\item[(1)]
\begin{eqnarray}\label{rt}
\Pi ^{\perp}(A_1,...,A_i)a^{\ell}\Pi (A_1,...,A_{i})=0.
\end{eqnarray}
\item[(2)]
\begin{eqnarray}\label{19c}
\varpi(A_{n+1}|A_1,...,A_{n})a^{\ell}\varpi(A_{n+1}|A_1,...,A_{n})=
A_{n+1}^{\ell}\varpi (A_{n+1}|A_1,...,A_{n}).
\end{eqnarray}
\item[(3)]
\begin{eqnarray}\label{19n}
&&{\rm Tr}[a^{\ell}\varpi (A_n|A_1,...,A_{n-1})]=A_n^{\ell};\;\;\;\;\;{\rm Tr}[a^{\ell}\Pi (A_1,...,A_{n})]=\sum _{i=1}^n A_i^{\ell}.
\end{eqnarray}

\end{itemize}

\end{proposition}
\begin{proof}
\mbox{}
\begin{itemize}

\item[(1)]
We prove this inductively. 
It is easily seen that it is true for $i=1$. We assume that it is true for $i=n$, i.e., that 
\begin{eqnarray}\label{rtq}
\Pi ^{\perp}(A_1,...,A_n)a^{\ell} \Pi (A_1,...,A_{n})=0
\end{eqnarray}
and we will prove that it is true for $i=n+1$. 

From Eq.(\ref{alg1}) follows that
\begin{eqnarray}\label{1121a}
\Pi^{\perp}(A_1,...,A_{n+1})&=&\Pi ^{\perp}(A_1,...,A_{n})-\varpi(A_{n+1}|A_1,...,A_{n}).
\end{eqnarray}
Therefore we need to prove that
\begin{eqnarray}\label{prove}
&&-\varpi(A_{n+1}|A_1,...,A_{n})a^{\ell}\Pi (A_1,...,A_{n})+ \Pi ^{\perp}(A_1,...,A_{n})a ^{\ell}\varpi(A_{n+1}|A_1,...,A_{n})\nonumber\\&&-\varpi(A_{n+1}|A_1,...,A_{n})a^{\ell}\varpi(A_{n+1}|A_1,...,A_{n})=0.
\end{eqnarray}
Firstly we point out that Eq.(\ref{rtq}) leads to
\begin{eqnarray}\label{zx1}
&&\varpi(A_{n+1}|A_1,...,A_{n})a^{\ell}\Pi (A_1,...,A_{i})\nonumber\\&&=\tau_{n+1} \Pi^{\perp}(A_1,...,A_{n}) \Pi(A_{n+1})\Pi ^{\perp}(A_1,...,A_{n})a^{\ell}\Pi (A_1,...,A_{i})=0
\end{eqnarray}
where $\tau_{n+1}$ has been given in Eq.(\ref{alg10}).
Also Eq.(\ref{rtq}) leads to
\begin{eqnarray}\label{cvb}
\Pi ^{\perp}(A_1,...,A_n)a^{\ell}\Pi ^{\perp}(A_1,...,A_{n})=\Pi ^{\perp}(A_1,...,A_n)a^{\ell}
\end{eqnarray}
which we use to prove that
\begin{eqnarray}\label{zx2}
&&\Pi ^{\perp}(A_1,...,A_{n})a^{\ell}\varpi(A_{n+1}|A_1,...,A_{n})=\tau _{n+1}\Pi^{\perp}(A_1,...,A_{n}) a ^{\ell}\Pi^{\perp}(A_1,...,A_{n}) \Pi(A_{n+1})\Pi ^{\perp}(A_1,...,A_{n})\nonumber\\&&=
\tau _{n+1}\Pi^{\perp}(A_1,...,A_{n}) a ^{\ell}\Pi(A_{n+1})\Pi ^{\perp}(A_1,...,A_{n})= A_{n+1}^{\ell}\varpi (A_{n+1}|A_1,...,A_{n})
\end{eqnarray}
We also use Eq.(\ref{cvb}) to prove that
\begin{eqnarray}\label{zx3}
&&\varpi(A_{n+1}|A_1,...,A_{n})a^{\ell}\varpi(A_{n+1}|A_1,...,A_{n})\nonumber\\&&=\tau _{n+1}^2\Pi^{\perp}(A_1,...,A_{n}) \Pi(A_{n+1})\Pi ^{\perp}(A_1,...,A_{n})a^{\ell}
\Pi^{\perp}(A_1,...,A_{n}) \Pi(A_{n+1})\Pi ^{\perp}(A_1,...,A_{n})\nonumber\\&&=\tau _{n+1}^2\Pi^{\perp}(A_1,...,A_{n}) \Pi(A_{n+1})\Pi ^{\perp}(A_1,...,A_{n})a^{\ell}
 \Pi(A_{n+1})\Pi ^{\perp}(A_1,...,A_{n})\nonumber\\&&=
A_{n+1}^{\ell}\varpi (A_{n+1}|A_1,...,A_{n}).
\end{eqnarray}
From Eqs(\ref{zx1}),(\ref{zx2}),(\ref{zx3}), follows Eq.(\ref{prove}). 
This proves Eq.(\ref{rt}), and it also proves Eqs(\ref{zx3}),(\ref{cvb}) (which are used below).
\item[(2)]
This has been proved in Eq.(\ref{zx3}).

\item[(3)]
Using Eq.(\ref{cvb}) we get
\begin{eqnarray}
&&{\rm Tr}[a^{\ell}\varpi (A_n|A_1,...,A_{n-1})]=\tau _n {\rm Tr}[a^{\ell}\Pi^{\perp}(A_1,...,A_{n-1}) \Pi(A_{n})\Pi ^{\perp}(A_1,...,A_{n-1})]\nonumber\\&&=
\tau _n {\rm Tr}[\Pi ^{\perp}(A_1,...,A_{n-1})a^{\ell}\Pi^{\perp}(A_1,...,A_{n-1}) \Pi(A_{n})]\nonumber\\&&=\tau _n {\rm Tr}[\Pi ^{\perp}(A_1,...,A_{n-1})a^{\ell}\Pi(A_{n})]
=A_n^{\ell}{\rm Tr}[\varpi (A_n|A_1,...,A_{n-1})]=A_n^{\ell}
\end{eqnarray}
Then we use this result and Eq.(\ref{59}) to prove the second of Eqs(\ref{19n}).

\end{itemize}
\end{proof}

We next use Eq.(\ref{19n}) and calculate the quantities
\begin{eqnarray}
&&{\rm Tr}[x\Pi (A_1,...,A_{n})]=\sqrt{2}\Re\left (\sum A_i\right )\nonumber\\
&&{\rm Tr}[p\Pi (A_1,...,A_{n})]=\sqrt{2}\Im\left (\sum A_i\right )
\end{eqnarray}
They are generalizations of similar results for coherent states.
We also use Eq.(\ref{35}) to get
\begin{eqnarray}
{\rm Tr}[a^{\dagger}a\Pi (A_1,A_2)]&=&|A_1|^2+|A_2|^2+{\mathfrak s}(|A_1-A_2|)
\end{eqnarray}
where
\begin{eqnarray}
{\mathfrak s}(|A|)&=&\frac{|A|^2}{\exp (|A|^2)-1}=\frac{1}{1+\frac{|A|^2}{2!}+\frac{|A|^4}{3!}+...};\;\;\;\;0< {\mathfrak s}(|A|)<1
\end{eqnarray}
In the limit $|A|\;\rightarrow\;0$ we get ${\mathfrak s}(|A|)=1$, and in the limit $|A|\;\rightarrow\;\infty$ we get ${\mathfrak s}(|A|)=0$.

If $\rho _0$ is a density matrix, we use the projectors  $\Pi(A_1,...,A_{n})$ to define generalized $Q$-functions as
\begin{eqnarray}
Q(A_1,...,A_{n})={\rm Tr}[\Pi(A_1,...,A_{n})\rho _0].
\end{eqnarray}
Using Eq.(\ref{985A}) we prove that
\begin{eqnarray}\label{9850}
\int _{\mathbb C}\frac{d^2A}{n\pi}Q(A,A+d_2,...,A+d_n)={1}.
\end{eqnarray}
In the case of the harmonic oscillator Hamiltonian $a^{\dagger}a$, the density matrix evolves in time as
\begin{eqnarray}
\rho (t)=\exp(ita^{\dagger}a)\rho _0 \exp (-ita^{\dagger}a)
\end{eqnarray}
and the corresponding $Q$-function, as a function of time is (Eq.(\ref{1121A}))
\begin{eqnarray}
{\rm Tr}[\Pi(A_1,...,A_{n})\exp(ita^{\dagger}a )\rho _0\exp(-ita^{\dagger}a )]=Q[A_1 \exp(-it),...,A_{n}\exp(-it)].
\end{eqnarray}
Also if we act with the displacement operators on both side of the density matrix, we get the density matrix $D(z) \rho _0 [D(z)]^{\dagger}$,
and the corresponding $Q$-function, as a function is (Eq.(\ref{1121}))
\begin{eqnarray}
{\rm Tr}[\Pi(A_1,...,A_{n})D(z) \rho _0 D^{\dagger}(z)]=Q(A_1-z,...,A_{n}-z).
\end{eqnarray}

\section{Coherent projectors in the Dirac contour representation}

The following proposition is  useful in practical calculations that involve coherent projectors.
\begin{proposition}\label{pro10}
If $S=\{A_1,...,A_i\}$, the $\Pi(S)$ is represented in the Dirac contour representation by the function:
\begin{eqnarray}\label{4zx}
&&\Pi(A_1,...,A_i)\;\rightarrow\;
\Theta(z_1,z_2)=\sum _{j,k}G_{jk}(S)\frac{\exp \left (A_jz_1-\frac{1}{2}|A_j|^2-\frac{1}{2}|A_k|^2\right )}{z_2-A_k^*}\nonumber\\
&&k,j=1,...,i;\;\;\;|z_2|>\max(|A_1|,...,|A_i|).
\end{eqnarray}
Here $G (S)$ is the inverse of the matrix $g(S)$, given in Eq.(\ref{pd}).
The set of poles of $\Theta(z_1,z_2)$ with respect to the variable $z_2$ is $S^*$, and the number $i$ of poles, is equal to the rank of $\Pi(S)$.
\end{proposition}
\begin{proof}
We first prove inductively that
\begin{eqnarray}\label{309}
\Pi(A_1,...,A_{j})\;\rightarrow\;\sum _{k=1}^{j}\frac{f_k(z_1)}{z_2-A_k^*}.
\end{eqnarray}
Here the $f_k(z_1)$ depend on $A_1,...,A_{j}$, but for simplicity we do not show this in the notation.

We have seen in Eq.(\ref{26}) that the $\Pi(A_1)$ is represented by such a function.
We assume that this is true for $j=i-1$ and we will prove that it is true for $j=i$.
 
From Eq.(\ref{309}) with $j=i-1$ (and Eq.(\ref{26}) for the operator ${\bf 1}$), it follows that 
\begin{eqnarray}\label{319}
\Pi^{\perp}(A_1,...,A_{i-1})\;\rightarrow\;
\frac{1}{z_2-z_1}-\sum _{k=1}^{i-1}\frac{f_k(z_1)}{z_2-A_k^*}.
\end{eqnarray}
Using Eq.(\ref{alg1}) in conjuction with Eq.(\ref{product}) we get
\begin{eqnarray}\label{alg11}
&&\varpi (A_i |A_1,...,A_{i-1})=\tau _i\Pi  ^{\perp} (A_1,...,A_{i-1}) \Pi (A_i)\Pi ^{\perp}(A_1,...,A_{i-1})
\nonumber\\&&\;\rightarrow\;\tau _i
\oint _C\frac{dz_2}{2\pi i} \oint _C\frac{dz_3}{2\pi i}\left (\frac{1}{z_2-z_1}-\sum _{k=1}^{i-1}\frac{f_k(z_1)}{z_2-A_k^*}\right)
\frac{\exp\left (A_iz_2-|A_i|^2\right )}{z_3-A_i^*}
\left (\frac{1}{z_4-z_3}-\sum _{k=1}^{i-1}\frac{f_k(z_3)}{z_4-A_k^*}\right)
\end{eqnarray}
Although there are many terms in these contour integrals, it is straightforward to see that the result is of the form
 $f_k(z_1)/(z_4-A_k^*)$ with $k=1,...,i-1$ and also with $k=i$ (this last term does not appear in Eq.(\ref{319})).
Adding these terms to Eq.(\ref{309}) with $j=i-1$, proves that Eq.(\ref{309}) also holds for $j=i$.

In order to prove Eq.(\ref{4zx}), we use the fact that
\begin{eqnarray}
\Pi(A_1,...,A_{i})\ket{A_j}=\ket{A_j};\;\;\;\;j=1,...,i.
\end{eqnarray}  
Using Eq.(\ref{1959}) we write this as
\begin{eqnarray}\label{3090}
\sum _k\oint _C \frac{dz_2 }{2\pi i} \frac{f_k(z_1)}{z_2-A_k^*}\exp(A_jz_2 -\frac{1}{2}|A_j|^2)=\exp(A_jz_1-\frac{1}{2}|A_j|^2).
\end{eqnarray}
From this follows that
\begin{eqnarray}
\sum _{k=1}^if_k(z_1)\exp \left(\frac{1}{2}|A_k|^2\right)g_{kj}(S)=\exp \left (A_jz_1-\frac{1}{2}|A_j|^2\right ).
\end{eqnarray}
This is a system of $i$ equations with $i$ unknowns, which gives 
\begin{eqnarray}\label{qzx}
f_k(z_1)=\sum _{j=1}^i\exp \left (A_jz_1-\frac{1}{2}|A_j|^2-\frac{1}{2}|A_k|^2\right )G_{jk}(S).
\end{eqnarray}
\end{proof}

\begin{cor}
Let $S_1=\{A_1,...,A_j\}$ and $S_2=\{B_1,...,B_i\}$ (in general $S_1\cap S_2\ne\emptyset$).
Also let
\begin{eqnarray}
&&\Pi(S_1)\;\rightarrow\;\sum _{r,s=1}^jG_{rs}(S_1)\frac{\exp \left (A_rz_1-\frac{1}{2}|A_r|^2-\frac{1}{2}|A_s|^2\right )}{z_2-A_s^*}
;\nonumber\\&&
\Pi(S_2)\;\rightarrow\;\sum _{\ell , m=1}^iG_{\ell m}(S_2)\frac{\exp \left (B_{\ell}z_1-\frac{1}{2}|B_{\ell}|^2-\frac{1}{2}|B_m|^2\right )}{z_2-B_m ^*}
\end{eqnarray}
be coherent projectors. 
The trace of the product of these coherent projectors, is given by
\begin{eqnarray}
{\rm Tr}[\Pi(S_1)\Pi(S_2)]=\sum _{r,s,\ell,m}G_{rs}(S_1)G_{\ell m}(S_2)\exp \left (A_rB_m^*+A_s^*B_\ell-\frac{1}{2}|A_r|^2-\frac{1}{2}|A_s|^2-\frac{1}{2}|B_{\ell}|^2-\frac{1}{2}|B_m|^2\right ).
\end{eqnarray}
\end{cor}

We next consider operators of the type described in Eq.(\ref{4zx}), but we 
replace the $G_{jk}(S)$ with arbitrary coefficients. We show 
that they `live\rq{} entirely within the space $H(A_1,...,A_{i})$, in the sense of the following proposition:
\begin{proposition}
Let $\Theta$ be operators which in the Dirac contour representation are represented by the functions 
\begin{eqnarray}
\Theta (z_1,z_2)=\sum _{j,k}\lambda_{jk}\frac{\exp \left (A_jz_1-\frac{1}{2}|A_j|^2-\frac{1}{2}|A_k|^2\right )}{z_2-A_k^*};\;\;\;j,k=1,...,i.
\end{eqnarray}
where $\lambda _{jk}$ are arbitrary complex numbers. Then 
\begin{eqnarray}\label{5cv}
\Pi(A_1,...,A_{i})\Theta \Pi(A_1,...,A_{i})=\Theta.
\end{eqnarray}
We denote the set of these operators as ${\cal A}(A_1,...,A_{i})$ (or ${\cal A}(S)$ with
$S=\{ A_1,...,A_{i}\}$).
\end{proposition}
\begin{proof}
These operators are
\begin{eqnarray}
\Theta (z_1,z_2)\;\leftrightarrow \;\Theta=\sum \lambda_{jk}\ket{A_j}\bra{A_k}.
\end{eqnarray}
Therefore $\Theta$ obeys Eq.(\ref{5cv}).
An alternative proof that uses the Dirac contour representation of these operators and Eq.(\ref{prod}), can also be given.  
\end{proof}
\begin{remark}
More general operators
\begin{eqnarray}
\Theta (z_1,z_2)=\sum _{j,k}\lambda_{jk}\frac{f_k(z_1)}{z_2-A_k^*};\;\;\;j,k=1,...,i,
\end{eqnarray}
do not belong in general to ${\cal A}(A_1,...,A_{i})$, and 
do not obey Eq.(\ref{5cv}). We exemplify this with functions with one pole
\begin{eqnarray}
\Theta=\ket{N}\bra{A}\;\leftrightarrow \;\frac{\exp (-\frac{1}{2}|A|^2)}{\sqrt{N!}}\frac{z_1^N}{z_2-A}.
\end{eqnarray}
In this case $\Pi(A)\Theta \Pi(A)$ is not equal to $\Theta$.
\end{remark}

\subsection{The non-orthogonal basis of coherent states in $H(S)$}

Let $S=\{ A_1,...,A_{i}\}$.
The coherent states $\ket{A_1}, ...,\ket{A_i}$ are a non-orthogonal basis in the $i$-dimensional space $H(S)$, and the
$g_{jk}(S)$ ($G_{jk}(S)$) are the metric (inverse of the metric) associated with this basis. 
An arbitrary state $\ket{s}$ is $H(S)$, can be expanded {\bf uniquely} in this basis as
\begin{eqnarray}\label{123}
\ket{s}=\sum _{j}s_j\ket{A_j};\;\;\;\;s_j=\sum _{k}G_{jk}(S)\bra{A_k}s\rangle.
\end{eqnarray}
Eq.(\ref{4zx}) can be rewritten as
\begin{eqnarray}\label{1qa}
\Pi(S)=\sum _{j,k}G_{jk}(S)\ket{A_j}\bra{A_k}.
\end{eqnarray}
In the following we represent quantum states with their components in the non-orthogonal basis of coherent states,
and operators with their matrix elements in this basis.
This means that operations between them will involve the matrices $g,G$ and the following lemma
summarizes the main relations that we need later.
\begin{lemma}
Let $\Theta$, $\Phi$ be operators in ${\cal A}(S)$ (i.e., $\Pi(S)\Theta \Pi(S)=\Theta$, and similarly for $\Phi$).
Also let
\begin{eqnarray}
\Theta _{k\ell}=\bra{A_k}\Theta \ket{A_{\ell}};\;\;\;\Phi _{k\ell}=\bra{A_k}\Phi \ket{A_{\ell}}\nonumber\\
\end{eqnarray}
Then:
\begin{itemize}
\item[(1)]
$\Theta$ is related to  $\Theta _{k\ell}$, as follows:
\begin{eqnarray}
\Theta=\sum _{j,k,m,n}G_{km}(S){\Theta}_{mn}G_{n\ell}(S)\ket {A_k}\bra{A_\ell}
\end{eqnarray}
\item[(2)]
In the non-orthogonal basis of coherent states,  the vector $\ket{t}=\Theta \ket{s}$ (where $\ket{s}$ is given in Eq.(\ref{123}), is represented by its components
\begin{eqnarray}\label{bbb}
t_j=\sum _{k, \ell}G_{jk}(S)\Theta _{k\ell}s_{\ell}.
\end{eqnarray}
\item[(3)]
The product $\Theta \Phi$ is represented with its matrix elements as:
\begin{eqnarray}\label{150}
(\Theta \Phi)_{k\ell}=\sum _{m,n}\Theta _{km}G_{mn}(S)\Phi _{n\ell},
\end{eqnarray}
and the commutator as
\begin{eqnarray}\label{153}
[\Theta ,\Phi]_{k\ell}=\sum _{m,n}\Theta _{km}G_{mn}(S)\Phi _{n\ell}-\sum _{m,n}\Phi _{km}G_{mn}(S)\Theta _{n\ell}.
\end{eqnarray}
\item[(4)]
If $\Theta$ is a unitary matrix, the unitarity relation $\Theta \Theta ^{\dagger}={\bf 1}$, is expressed as
\begin{eqnarray}\label{uni}
\sum _{m,n}\Theta _{km}G_{mn}(S)(\Theta ^{\dagger}) _{n\ell}=g_{k \ell}(S).
\end{eqnarray}
\item[(5)]
The projector $\Pi(S)$, which can be viewed as the unit operator acting only on $H(S)$, has elements
\begin{eqnarray}\label{511}
{[\Pi(S)]}_{k\ell}=\bra{A_k}{\Pi(S)}\ket{A_l}=g_{k\ell}(S). 
\end{eqnarray}
\end{itemize}
\end{lemma}
\begin{proof}
The proof of all these relations is based on Eq.(\ref{1qa}).
\end{proof}
This formalism is in the spirit of the Berezin\cite{BERE} covariant and contravariant symbols.
We can introduce dual quantities as
\begin{eqnarray}
\Theta=\sum _{j,k}{\widetilde \Theta}_{jk}\ket {A_j}\bra{A_k}.
\end{eqnarray}
The ${\widetilde \Theta}_{jk}$ are related to ${\Theta}_{jk}$ as
\begin{eqnarray}
\Theta _{k\ell}=\sum _{m,n}g_{km}(S){\widetilde \Theta}_{mn}g_{n\ell}(S) ;\;\;\;\;
{\widetilde \Theta}_{k\ell}=\sum _{m,n}G_{km}(S){\Theta}_{mn}G_{n\ell}(S),
\end{eqnarray}
but we will not use them.
A notation with upper and lower indices (and the rule that we only contract 
a lower index with an upper index, and we do not contract two lower indices or two upper indices) would `hide' the $g,G$ matrices, but we do not use it here.

\section{The Boolean ring of finite sets and classical gates}\label{X2}

In this section we consider finite sets of complex numbers, 
which in our context are finite sets of poles, and are related to coherent subspaces.
Following Stone \cite{S,S1,JS}, we study their structure from a lattice theory point of view, and
from a ring theory point of view.
Later in section \ref{X1}, all these results are transfered isomorphically to the coherent subspaces.

We show that the set of finite sets of complex numbers, is a distributive lattice $\cal L$.
We then prove that $\cal L$ can be described as particular type of ring, known as a Boolean ring.
In the case of distributive lattices which are Boolean algebras, the corresponding Boolean rings have identity.
But in our case the distributive lattice $\cal L$ is not a Boolean algebra, and the Boolean ring does not have identity.
For this reason, we use the weaker definition of a ring, which does not require the existence of identity.
However we also consider sublattices of $\cal L$, which are Boolean algebras and they do have identity.

The formalism provides the theoretical foundation for the study of classical gates and also quantum gates with coherent states.

\subsection{The distributive lattice $\cal L$}
We consider the set $\cal L$ of all finite subsets of ${\mathbb C}$ (the empty set $\emptyset$ is an element of $\cal L$).
For $S_1, S_2 \in \cal L$, we define the subset, the union and the intersection as 
the partial order $\prec$, disjunction $\vee$, and conjunction $\wedge$, correspondingly:
\begin{eqnarray}\label{mmm}
S_1\prec S_2\;\leftrightarrow\;S_1\subset S_2;\;\;\;\;
S_1\vee S_2=S_1\cup S_2;\;\;\;\;\;
S_1\wedge S_2=S_1\cap S_2
\end{eqnarray}
The symbols $\subset$ and $\prec$, include equality.
These operations are performed only a finite number of times. Then 
$\cal L$ is closed under these operations, and it is a distributive lattice.
$\cal L$ has $0$ (least element) which is the empty set $\emptyset$.
$\cal L$ does not have $1$ (greatest element), and we can not define complements (the $\mathbb C \setminus S$ does not belong to ${\cal L}$). 
We can define the relative complement of a set $S$ in an interval $\emptyset \subset S\subset R$, with respect to $R$, which is the
${\overline S}=R\setminus S$. Clearly $\cal L$ is not a Boolean algebra.

We have seen in definition \ref{defi} that the coherent bra functions have poles in a finite set $R\subset {\mathbb C}$ or a subset of $R$.
This motivates the study of all subsets of $R$, and leads to the concept of a principal ideal in $\cal L$, which is 
 the powerset (set of subsets) of $R$:
\begin{eqnarray}\label{aqq}
{\cal I}(R)=\{S\in{\cal L}\;|\;S\subset R\}
\end{eqnarray}
The cardinality of ${\cal I}(R)$ is $2^{|R|}$.
${\cal I}(R)$ is a Boolean algebra with the set $R$ as $1$.
The complement of an element $S$ is $\overline S=R\setminus S$.

\subsection{$\cal L$ as a Boolean ring}\label{boole}

In the set  $\cal L$ of finite sets $S$ of complex numbers, we define multiplication as intersection, and addition as symmetric difference:
\begin{eqnarray}
&&S_1+S_2=(S_1\setminus S_2)\cup (S_2\setminus S_1)\nonumber\\
&&S_1\cdot S_2=S_1\cap S_2
\end{eqnarray}

The $S_1\cdot S_2=S_1\cap S_2$ is the logical AND operation, the $S_1\cup S_2$ is the logical OR operation,
and the $S_1+S_2$ is the logical XOR (excluded OR) operation.
In this section we have replaced the $S_1\cup S_2$ (logical OR) with the  $S_1+S_2$ (logical XOR).
The merit of doing this, is that we get a Boolean ring, with the properties given below.
The $S_1\cup S_2$ can be expressed in terms of addition and multiplication, as
\begin{eqnarray}\label{73}
S_1\cup S_2=S_1+S_2+(S_1\cdot S_2).
\end{eqnarray}
Only finite sums and finite products are considered and then $\cal L$ is closed under multiplication and addition.
Both addition and multiplication are commutative and associative.
Also distributivity  holds:
\begin{eqnarray}\label{A1}
S_1\cdot (S_2+S_3)=(S_1\cdot S_2)+(S_1\cdot S_3)
\end{eqnarray}
The $\emptyset$ plays the role of additive zero. The additive inverse of a set is the set itself ($S_1=-S_1$) and therefore
\begin{eqnarray}
S_2-S_1=S_2+S_1=S_1-S_2=-S_1-S_2.
\end{eqnarray}
From Eq.(\ref{73}) it follows that
\begin{eqnarray}\label{74}
S_1+S_2=(S_1\cup S_2)-(S_1\cdot S_2)=(S_1\cup S_2)+(S_1\cdot S_2).
\end{eqnarray}
It is easily seen that
\begin{eqnarray}\label{A2}
S_1+\emptyset=S_1;\;\;\;S_1+S_1=\emptyset;\;\;\;\;S_1\cdot S_1=S_1.
\end{eqnarray}
The easiest way to prove these properties is using Venn diagrams.
The multiplication is idempotent.
Therefore $\cal L$ is a commutative ring with the extra property of idempotent multiplication.
The following relation also holds:
\begin{eqnarray}\label{A3}
&&(S_1\cdot S_2)\cdot (S_1+S_2)=\emptyset.
\end{eqnarray}
Therefore all the elements of this ring, are divisors of zero.
Also
\begin{eqnarray}\label{ineq}
S_1\prec S_2\;\rightarrow\;S_1\cdot S_0\prec S_2\cdot S_0.
\end{eqnarray}
But $S_1\prec S_2$ does not imply $S_1+ S_0\prec S_2+ S_0$.

A ring which has idempotent multiplication is commutative, and it is called Boolean ring \cite{S,JS}.
Boolean rings with an identity are Boolean algebras. In our case ${\cal L}$ does not have an identity, and it is not a Boolean algebra.

It is easily seen that the ideals ${\cal I}(R)$ defined in Eq.(\ref{aqq}) within lattice theory, are also ideals within ring theory (with the addition and multiplication of their elements defined above). In fact they
are Boolean rings with the set $R$ as identity, i.e., they are Boolean algebras. 
The complement of an element  $S\in {\cal I}(R)$, is ${\overline S}=S+R=R\setminus S$.

\subsection{Reversible classical gates and the CNOT gate}\label{VV}

A classical gate is a function ${\cal M}$ that maps an input $(A_1,...,A_n)$ to an output $(B_1,...,B_m)$:
\begin{eqnarray}
{\cal M}(A_1,...,A_n)=(B_1,...,B_m)
\end{eqnarray}
In most of the literature the input and output  variables $A_i,B_j$, are binary.
The case where they take $d$ values has also been studied, but to a lesser extent (e.g, \cite{D1}
in a classical context, and \cite{D2,D3} in a quantum context). 

In our case the inputs and outputs $A_i,B_j$ are 
finite subsets of ${\mathbb C}$ which are  
elements of a principal ideal ${\cal I}(R)$.
Therefore a gate is a function $\cal M$ from $[{\cal I}(R)]^n$ to $[{\cal I}(R)]^m$ (where $[{\cal I}(R)]^n$ is the Cartesian product of $n$ of these sets).
If $R$ has cardinality $|R|=1$, then the inputs and outputs are binary variables (the $\emptyset$, $R$ which can be represented with $0,1$).
For  $|R|\ge 2$,
the inputs and outputs take one of $d=2^{|R|}$ values, and they can be represented with $0,...,2^{|R|}-1$.
This leads to generalizations of the classical gates with binary variables, to gates with $2^{|R|}$-ary variables.

Examples of gates are the OR, AND, XOR (from $[{\cal I}(R)]^2$ to ${\cal I}(R)$), and the NOT 
(from ${\cal I}(R)$ to ${\cal I}(R)$):
\begin{eqnarray}
&&{\cal M}_{\rm OR}(S_1,S_2)=S_1+S_2+S_1\cdot S_2=S_1\vee S_2;\;\;\;\;S_1,S_2\in {\cal I}(R)\nonumber\\
&&{\cal M}_{\rm AND}(S_1,S_2)=S_1\cdot S_2=S_1\wedge S_2\nonumber\\
&&{\cal M}_{\rm XOR}(S_1,S_2)=S_1+S_2\nonumber\\
&&{\cal M}_{\rm NOT}(S_1)=R+S_1={\overline S_1}=R\setminus S_1.
\end{eqnarray}
The sets are labels for the `real\rq{} inputs and outputs, which for classical gates are electric currents with various values.
As a `non-binary example\rq{}, we consider the OR, AND, XOR gates with $R=\{A_1,A_2\}$ and for convenience use the notation:
\begin{eqnarray}\label{or}
\emptyset\;\;\rightarrow\;\;0;\;\;\;\;\{A_1\}\;\;\rightarrow\;\;1;\;\;\;\;\{A_2\}\;\;\rightarrow\;\;2
;\;\;\;\;\{A_1,A_2\}\;\;\rightarrow\;\;3
\end{eqnarray}
The outputs of these gates are shown in table \ref{t1}.

In a reversible classical gate, there is a bijective map between the set of all inputs and the set of all outputs.
In other words, to every output corresponds exactly one input.
The OR, AND, XOR,  are not reversible gates, but the NOT is a reversible gate.
We are interested in  reversible classical gates, because
unitary quantum transformations are reversible, and in this sense reversible classical gates are linked to quantum gates.
As an example, we study the classical CNOT  (controlled NOT) gate \cite{G1,G2,G3,G4}, with binary and more generally 
$2^{|R|}$-ary variables, using the language of Boolean rings, discussed earlier.

The CNOT gate is a bijective function from $[{\cal I}(R)]^2$ to itself:
\begin{eqnarray}\label{cnot}
{\cal M}(S_1,S_2)=(S_1,S_1+S_2)
\end{eqnarray}
$S_1$ is the `control input\rq{} and $S_2$ is the `target input\rq{}.
We note that:
\begin{itemize}
\item
If the control input is $S_1=\emptyset$, the target input remains unchanged:
\begin{eqnarray}
{\cal M}(\emptyset ,S_2)=(S_1,S_2)
\end{eqnarray}
\item
If the control input is $S_1=R$ (the identity in ${\cal I}(R)$), the target changes from 
$S_2$ to its complement $\overline {S_2}=R+S_2=R\setminus S_2$:
\begin{eqnarray}
&&{\cal M}(R,S_2)=(R, \overline {S_2}).
\end{eqnarray}
\item
If $S_1 \subset S_2$, the target changes from 
$S_2$ to its subset ${S_2}\setminus S_1$.
\item
If $S_2 \subset S_1$, the target changes from 
$S_2$ to ${S_1}\setminus S_2$.
\end{itemize}
For $R=\{A_1\}$ (i.e., $|R|=1$) this is the CNOT gate, with binary variables.
Using the notation
\begin{eqnarray}\label{oror}
\emptyset\;\;\rightarrow\;\;0;\;\;\;\;\{A_1\}\;\;\rightarrow\;\;1
\end{eqnarray}
we give the $4$ possible inputs, and the corresponding outputs in table \ref{t2}.
As a `non-binary example\rq{}, we consider the case  where $R=\{A_1, A_2\}$.
In this case we have $16$ possible inputs, and the corresponding outputs are shown in table \ref{t3}
(using the notation in Eq.(\ref{or})). 

The following proposition is based heavily on
the formalism of Boolean rings. In this sense it translates the structure of Boolean rings into the language of classical CNOT gates. 
It will also be expressed later, in the language of quantum CNOT gates.

\begin{proposition}\label{111}
For a fixed control input $S_1\in {\cal I}(R)$, the map from the target input to the target output 
\begin{eqnarray}\label{112}
{\cal M}_{jT}(S_2)=S_1+S_2;\;\;\;\;j=1,...,2^{|R|}
\end{eqnarray}
is a bijective map from ${\cal I}(R)$ to itself.
Also ${\cal M}_{jT}\circ {\cal M}_{jT}={\bf 1}$.
\end{proposition}
\begin{proof}
${\cal I}(R)$ is a Boolean ring, and therefore it is an Abelian group with respect to addition. 
Using this we prove that ${\cal M}_{jT}$ is a bijective map from ${\cal I}(R)$ to itself.
For example if ${\cal M}_{jT}(S_1)={\cal M}_{jT}(S_1')$, the group properties prove that $S_1=S_1'$.

The property $S_1+S_1=0$, implies that ${\cal M}_{jT}\circ {\cal M}_{jT}={\bf 1}$.
\end{proof}
Tables \ref{t2}, \ref{t3} provide examples of this proposition.

\section{The Boolean ring of coherent spaces and quantum gates}\label{X1}
\subsection{The  distributive lattice ${\cal L} _{\rm coh}\simeq {\cal L}$ }\label{p1}

${\cal L}_{\rm coh}$ is defined to be the set of coherent subspaces $H(S)$, with $S$ a finite subset of ${\mathbb C}$.
The $H(\emptyset)={\cal O}$ is an element of ${\cal L}_{\rm coh}$.
Let $S_1,S_2$ be finite subsets of ${\mathbb C}$.
The set ${\cal L}_{\rm coh}$, with the disjunction $H(S_1)\vee H(S_2)$ (Eq.(\ref{1})), conjunction $H(S_1)\wedge H(S_2)$
(Eq.(\ref{2})), and subspace partial order
$H(S_1)\prec H(S_2)$, is a distributive lattice which we call `coherent lattice\rq{}.
Only a finite number of disjunctions and conjunctions are considered, so that ${\cal L}_{\rm coh}$ is closed under these operations.

For coherent subspaces, we can prove that 
\begin{eqnarray}\label{ty}
H(S_1)\vee H(S_2)=H(S_1 \cup S_2);\;\;\;\;\;H(S_1)\wedge H(S_2)=H(S_1\cap S_2).
\end{eqnarray} 
The proof that $H(S_1)\wedge H(S_2)=H(S_1)\cap H(S_2)$ is equal to $H(S_1\cap S_2)$
is based on the fact that a finite number of coherent states, is linearly independent.

The ${\cal O}$ is the zero in this lattice.
There is no $1$ in this lattice (the full Hilbert space $h$ does not belong to ${\cal L}_{\rm coh}$).
The negation operation is not defined in this lattice, and the orthocomplements of coherent projectors are not included in 
${\cal L} _{\rm coh}$.
Comparison of Eqs.(\ref{mmm}),(\ref{ty}) shows that the lattice ${\cal L}_{\rm coh}$ is isomorphic to the 
distributive lattice ${\cal L}$.

In analogy with  Eq.(\ref{aqq}), the principal ideal  of all coherent subspaces of the coherent space $H(R)$ is:
\begin{eqnarray}\label{wq}
{\cal I}_{\rm coh}(R)=\{H(S)\in{\cal L} _{\rm coh}\;|\;S\subset R\}.
\end{eqnarray}

\paragraph*{The quantum OR in terms of  measurements:} We perform a measurement with the projector $\Pi(S_i)$ ($i=1,2$)
 on a density matrix $\rho$. If the outcome is `yes\rq{}, then the system collapses in the state
 \begin{eqnarray}
\rho\rq{}=\frac{\Pi(S_i)\rho \Pi(S_i)}{{\rm Tr}[\rho \Pi(S_i)]}.
\end{eqnarray}
A subsequent  measurement on this with the projector $\Pi (S_1 \cup S_2)$ (which detects if the system is in $H(S_1)\vee H(S_2)$)
will give `yes\rq{} with probability $1$.

\paragraph*{The quantum AND in terms of  measurements:}
We perform a measurement with the projector $\Pi(S_1\cap S_2)$
 on a density matrix $\rho$. If the outcome is `yes\rq{}, then the system collapses in the state
 \begin{eqnarray}
\rho\rq{}=\frac{\Pi (S_1\cap S_2)\rho \Pi(S_1\cap S_2)}{{\rm Tr}[\rho \Pi(S_1\cap S_2)]}.
\end{eqnarray}
A subsequent  measurement on this with any of the projector $\Pi (S_i )$
will give `yes\rq{} with probability $1$.

\subsection{${\cal L} _{\rm coh}$  as a Boolean ring}\label{p2}

In ${\cal L} _{\rm coh}$ we define addition and multiplication as:
\begin{eqnarray}\label{rrr}
H(S_1)+H(S_2)=H(S_1+ S_2);\;\;\;\;\;H(S_1)\cdot H(S_2)=H(S_1\cdot S_2)=H(S_1)\wedge H(S_2).
\end{eqnarray}
The $H(S_1)+H(S_2)$ is the logical XOR operation.
Only finite sums and finite products, are considered.

It is easily seen that
\begin{eqnarray}
H(S_1+ S_2)=H(S_1\setminus S_2)\vee H(S_2\setminus S_1).
\end{eqnarray}
We note that 
\begin{eqnarray}
H(S_1\setminus S_2)\wedge H(S_2\setminus S_1)={\cal O},
\end{eqnarray}
and in this sense, a vector which belongs entirely in $H(S_1\setminus S_2)$ does not belong in $H(S_2\setminus S_1)$, and vice versa.
However, the $H(S_1+ S_2)$ contains the vectors in the spaces $H(S_1\setminus S_2)$, $H(S_2\setminus S_1)$, and their superpositions.
In this sense, quantum XOR is different from classical XOR.

${\cal L}_{\rm coh}$ with these operations
is a commutative ring (without identity) and with idempotent multiplication:
\begin{eqnarray}\label{B1}
H(S_1)\cdot H(S_1)=H(S_1).
\end{eqnarray}
Therefore it is a Boolean ring, isomorphic to ${\cal L}$.
Properties analogous to those in Eq.(\ref{73})-(\ref{A3}) hold here also, and for convenience we summarize them:
\begin{eqnarray}\label{B1}
&&H(S_1)\vee H(S_2)=H(S_1)+H(S_2)+[H(S_1)\cdot H(S_2)]\nonumber\\
&&H(S_1)\cdot[ H(S_2)+ H(S_3)]=[H(S_1)\cdot H(S_2)]+[H(S_1)\cdot H(S_3)]\nonumber\\
&&H(S_1)+{\cal O}=H(S_1)\nonumber\\
&&H(S_1)+H(S_1)={\cal O};\;\;\;\;\;H(S_1)=-H(S_1)\nonumber\\
&&H(S_1)\cdot H(S_2)\cdot H(S_1+S_2)={\cal O}.
\end{eqnarray}
The ideals given in Eq.(\ref{wq}) from a lattice theory point of view, are also ideals from a ring theory point of view.
In analogy to Eq.(\ref{ineq})
\begin{eqnarray}
H(S_1)\prec H(S_2)\;\rightarrow\;H(S_1)\cdot H(S_0)\prec H(S_2)\cdot H(S_0).
\end{eqnarray}
But $H(S_1)\prec H(S_2)$ does not imply $H(S_1)+ H(S_0)\prec H(S_2)+ H(S_0)$.

\paragraph*{The quantum XOR in terms of  measurements:} We perform a measurement with the projector 
$\Pi(S_1\setminus S_2)$
 on a density matrix $\rho$. If the outcome is `yes\rq{}, then the system collapses in the state
 \begin{eqnarray}
\rho\rq{}=\frac{\Pi(S_1\setminus S_2)\rho \Pi(S_1\setminus S_2)}{{\rm Tr}[\rho \Pi(S_1\setminus S_2)]}.
\end{eqnarray}
A subsequent  measurement on this with the projector $\Pi (S_1 + S_2)$ (which detects if the system is in $H(S_1\setminus S_2)$ or in $H(S_2\setminus S_1)$)
will give `yes\rq{} with probability $1$.

\subsection{Quantum CNOT gates with input in the coherent space $H_A(A_1,A_2)\otimes H_B(B_1,B_2)$}\label{GG1}

In a controlled quantum gate \cite{G1,G2,G3,G4,G5}, we have the unitary transformations:
 \begin{eqnarray}
\ket{e}\otimes \ket{t}\;\rightarrow \;\ket{e}\otimes ({\cal U}_T\ket{t})
;\;\;\;\;\ket{e}\in h_1;\;\;\;\;\;\ket{t}\in h_2.
\end{eqnarray}
Here $\ket{e}$ is the quantum state in the control input, and $\ket{t}$ is the quantum state in the target input.
The control output is the same quantum state as in the control input.
The gate performs the unitary transformation ${\cal U}_T$ (the index T indicates `target\rq{}) on the target output. Depending on the control input, the ${\cal U}_T$ is ${\bf1}$, or another unitary operator.

Here we consider the quantum version of the CNOT gate in table \ref{t2}.
There are two possible inputs in the control and in the target, and therefore the required coherent space is
$H_A(A_1,A_2)\otimes H_B(B_1,B_2)$.
An arbitrary vector in the non-orthogonal `coherent basis' in this space, is
\begin{eqnarray}
\begin{pmatrix}
\alpha _1\\
\alpha _2
\end{pmatrix}
\otimes
\begin{pmatrix}
\beta _1\\
\beta _2
\end{pmatrix}\;\rightarrow\;
[\alpha _1\ket{A_1}+\alpha _2\ket{A_2}]\otimes [\beta _1\ket{B_1}+\beta _2\ket{B_2}]
\end{eqnarray}
The scalar product of such vectors will involve the $g(A_1,A_2)\otimes g(B_1,B_2)$ as in Eq.(\ref{prod}), and it gives the expected results
for overlaps of coherent states.

The CNOT gate performs the following unitary transformation (it is a sum of tensor products of $2\times 2$ `control matrices' with $2\times 2$ `target matrices' )
\begin{eqnarray}\label{339}
&&U=\gamma _{1A}E_1(A_1,A_2)\otimes {\cal U}_{1T}+\gamma _{2A}E_2(A_1,A_2)\otimes {\cal U}_{2T}\nonumber\\
&&{\cal U}_{1T}=V(0);\;\;\;{\cal U}_{2T}=V(\pi);\;\;\;\;
V(\phi)=\gamma _{1B}E_1(B_1,B_2)+\exp(i\phi)\gamma _{2B}E_2(B_1,B_2)
\end{eqnarray}
where $\gamma _{jA}, E_j(A_1,A_2)$ and $\gamma _{jB}, E_j(B_1,B_2)$ are the eigenvalues and eigenprojectors of 
$g(A_1,A_2), g(B_1,B_2)$, correspondingly (see example \ref{ex1}).We note that
\begin{eqnarray}\label{c45}
{\cal U}_{1T}= g(B_1,B_2);\;\;\;{\cal U}_{2T}= g(B_1,B_2)-2\gamma _{2B}E_2(B_1,B_2);\;\;\;[{\cal U}_{1T},{\cal U}_{2T}]=0.
\end{eqnarray}

We first show that $U$ is a unitary operator, by proving that
\begin{eqnarray}\label{cc45}
U[G(A_1,A_2)\otimes G(B_1,B_2)]U^{\dagger}=g(A_1,A_2)\otimes g(B_1,B_2).
\end{eqnarray}
This is the unitarity relation in Eq.(\ref{uni}) in the non-orthogonal basis of coherent states, extended for a tensor product.
We first show that
\begin{eqnarray}
U[G(A_1,A_2)\otimes G(B_1,B_2)]=E_1(A_1,A_2)\otimes {\bf 1}+E_2(A_1,A_2)\otimes [E_1(B_1,B_2)-E_2(B_1,B_2)],
\end{eqnarray}
and then multiplication with $U^{\dagger}$ gives the result in Eq.(\ref{cc45}).

We also show that ${\cal U}_{jT}$ are unitary transformations, by showing that they satisfy
the unitarity relation  (Eq.(\ref{uni})):
\begin{eqnarray}\label{c46}
{\cal U}_{jT}G(B_1,B_2)({\cal U}_{jT})^{\dagger}= g(B_1,B_2).
\end{eqnarray}

We consider two cases of linearly independent control inputs:
\begin{itemize}
\item
If the control input is in the state ${e_1(A_1,A_2)}$ in Eq.(\ref{pd15}), then
\begin{eqnarray}\label{AA}
U[{e_1(A_1,A_2)}\otimes {\cal T}]={e_1(A_1,A_2)}\otimes {\cal T};\;\;\;{\cal T}=
\begin{pmatrix}
t _1\\
t _2
\end{pmatrix}\;\rightarrow t_1\ket{B_1}+t_2\ket{B_2}.
\end{eqnarray} 
We use here Eq.(\ref{bbb}) (which involves 
the $G(A_1,A_2)\otimes G(B_1,B_2)$).
We also use the fact that ${\cal U}_{1T}=g(B_1,B_2)$, to show that ${\cal U}_{1T}{\cal T}={\cal T}$.
In this case the target output is the same as the target input.
This is the analogue of $(0,0)\rightarrow (0,0)$ and $(0,1)\rightarrow (0,1)$ in the classical CNOT gates (table \ref{t2}).
But in addition to that, here we have superpositions, i.e., the ${\cal T}$ is any vector in the two-dimensional space $H_B(B_1,B_2)$.
\item
If the control input is in the state ${e_2(A_1,A_2)}$ (Eq.(\ref{pd15})), 
then
\begin{eqnarray}
&&U[{e_2(A_1,A_2)}\otimes {\cal T}]={e_2(A_1,A_2)}\otimes ({\cal U}_{2T}{\cal T})
\end{eqnarray} 
In this case the target input ${\cal T}$ is transformed to ${\cal U}_{2T}{\cal T}$ at the output.
\end{itemize}
We next consider the states
\begin{eqnarray}
&&{\cal T}_p=\frac{1}{\sqrt 2}[{e_1(B_1,B_2)}+{e_2(B_1,B_2)}];\;\;\;\;{\cal T}_m=\frac{1}{\sqrt 2}[{e_1(B_1,B_2)}-{e_2(B_1,B_2)}].
\end{eqnarray} 
The quantum CNOT gate gives
\begin{eqnarray}
&&[e_1(A_1,A_2), {\cal T}_p]\;\;\rightarrow\;\;[e_1(A_1,A_2), {\cal T}_p];\;\;\;\;\;[e_1(A_1,A_2), {\cal T}_m]\;\;\rightarrow\;\;[e_1(A_1,A_2), {\cal T}_m]
\nonumber\\
&&[e_2(A_1,A_2), {\cal T}_p]\;\;\rightarrow\;\;[e_2(A_1,A_2), {\cal T}_m];\;\;\;\;\;[e_2(A_1,A_2), {\cal T}_m]\;\;\rightarrow\;\;[e_2(A_1,A_2), {\cal T}_p]
\end{eqnarray} 
This is the quantum analogue of table \ref{t2} in the classical CNOT gates.
Therefore we have here a quantum analogue of the map ${\cal M}_{jT}$ in proposition \ref{111} for the classical case, which satisfies the condition
${\cal M}_{jT}\circ {\cal M}_{jT}={\bf 1}$.
But in the quantum case we can also have superpositions (both in the target input and  in the control input), and for this reason we present 
the following more general statement, as the quantum analogue of proposition \ref{111} .
\begin{proposition}\label{111q}
For a fixed control input $e_j(A_1,A_2)$, the map from the target input to the target output 
\begin{eqnarray}\label{112}
{\cal M}_{jT}:\;\;{\cal T}\;\;\rightarrow\;\;{\cal U}_{jT}{\cal T}
\end{eqnarray}
is a bijective map  from $H_B(B_1,B_2)$ to itself.
Also ${\cal M}_{jT}\circ {\cal M}_{jT}={\bf 1}$.
\end{proposition}
\begin{proof}
We have shown in Eq.(\ref{c46}) that the  transformations ${\cal U}_{jT}$ are unitary, and therefore ${\cal M}_{jT}$ is a 
bijective map  from $H_B(B_1,B_2)$ to itself. We can also show that $({\cal U}_{jT})^2={\bf 1}$
(the square is calculated  using Eq.(\ref{150})).
Therefore ${\cal M}_{jT}\circ {\cal M}_{jT}={\bf 1}$.
\end{proof}

In the case $A_1=-A_2\rightarrow \infty$ and $B_1=-B_2\rightarrow \infty$ the coherent states are almost orthogonal and
\begin{eqnarray}
g(A_1,A_2)\approx g(B_1,B_2)\approx{\bf 1};\;\;\;\gamma _{1A}\approx \gamma _{2A}\approx \gamma _{1B}\approx \gamma _{2B}\approx 1
\end{eqnarray} 
and Eq.(\ref{333}) reduces to
\begin{eqnarray}
U\approx E_1(A_1,A_2)\otimes {\bf 1}+E_2(A_1,A_2)\otimes {\cal U}_{2T}
\end{eqnarray}
This is the form found in the literature (usually in another basis).

\subsection{Quantum CNOT gates with input in the coherent space $H_A(S_A)\otimes H_B(S_B)$}\label{GG2}

Here we consider the quantum version of the CNOT gate in table \ref{t3}.
There are four possible inputs in the control and in the target, and therefore the required coherent space is
$H_A(S_A)\otimes H_B(S_B)$, where $S_A=\{A_1,A_2,A_3,A_4\}$ and 
$S_B=\{B_1,B_2,B_3,B_4\}$.
The CNOT gate performs the following transformation on input
states (it is a sum of tensor products of $4\times 4$ `control matrices' with  $4\times 4$ `target matrices'):
\begin{eqnarray}\label{333}
U&=&\gamma _{1A}E_1(S_A)\otimes {\cal U}_{1T}+\gamma _{2A}E_2(S_A)\otimes {\cal U}_{2T}\nonumber\\
&+&\gamma _{3A}E_3(S_A)\otimes {\cal U}_{3T}+\gamma _{4A}E_4(S_A)\otimes {\cal U}_{4T}
\end{eqnarray}
where
\begin{eqnarray}
&&{\cal U}_{1T}=V(0,0,0);\;\;\;{\cal U}_{2T}=V(\pi,0,0);\;\;\;{\cal U}_{3T}=V(0,\pi,0);\;\;\;{\cal U}_{4T}=V(0,0,\pi)\nonumber\\
&&V(\phi _2, \phi _3,\phi_4)=\gamma _{1B}E_1(S_B)+\exp(i\phi _2)\gamma _{2B}E_2(S_B)\nonumber\\&&+
\exp(i\phi _3)\gamma _{3B}E_3(S_B)+\exp(i\phi _4)\gamma _{4B}E_4(S_B),
\end{eqnarray}
$U$ is a unitary operator. It obeys the unitarity relation (Eq.(\ref{uni})):
\begin{eqnarray}
U[G(S_A)\otimes G(S_B)]U^{\dagger}=g(S_A)\otimes g(S_B).
\end{eqnarray}
${\cal U}_{jT}$ is also a unitary operator:
\begin{eqnarray}
{\cal U}_{jT} G(S_B){\cal U}_{jT}^{\dagger}= g(S_B).
\end{eqnarray}
We note that
\begin{eqnarray}\label{c45}
&&{\cal U}_{1T}= g(S_B);\;\;\;\;{\cal U}_{2T}= g(S_B)-2\gamma _{2B}E_2(S_B)\nonumber\\
&&{\cal U}_{3T}= g(B_1,B_2)-2\gamma _{3B}E_3(S_B);\;\;\;\;{\cal U}_{4T}= g(B_1,B_2)-2\gamma _{4B}E_4(S_B)
\end{eqnarray}
If the control input is in the eigenstate ${e_j(S_A)}$, given in Eq.(\ref{eigv}), then
\begin{eqnarray}
U[{e_j(S_A)}\otimes {\cal T}]={e_j(S_A)}\otimes [{\cal U}_{jT}{\cal T}]
\end{eqnarray} 
In this case the target input $\cal T$ is transformed to ${\cal U}_{jT}\cal T$ at the output.
It is seen that for any of the control inputs $e_j(S_A)$, 
the target input ${\cal T}$ is transformed with the unitary operator ${\cal U}_{jT}$, into ${\cal U}_{jT}{\cal T}$.
This is a bijective map from $H(S_B)$ to itself, and a proposition analogous to \ref{111q}, holds here also.

\section{Discussion}
We have discussed the following three interrelated topics:
\begin{itemize}
\item
{\bf Coherent spaces:} They are subspaces of the Hilbert space, spanned by a finite number of coherent states. 
Each coherent space is described uniquely by a finite set of complex numbers.
Using the language of the Dirac contour representation, we have shown that
the corresponding projectors, have the following properties:
\begin{itemize}
\item
There is a resolution of the identity in terms of all them (Eq.(\ref{985A})). 
Also some smaller sets of coherent spaces are total sets, as discussed in section \ref{totalset}.
 \item
Under both displacement transformations and time evolution, they are transformed into other projectors of the same type (proposition \ref{propo1}).
\item
They obey the relations in proposition \ref{bnm}, which are extensions of the fact that coherent states are eigenstates of the annihilation operator. 
\end{itemize}

\item
{\bf The Boolean ring   of finite sets of complex numbers:}
The set ${\cal L}$ of all finite sets of complex numbers with the logical OR and AND operations, is a distributive lattice. It is also a Boolean ring, with the operations XOR, AND, and it has the properties discussed in section \ref{boole}.
The Boolean ring provides the theoretical foundation for a study of classical and quantum gates, 
with binary and more generally $2^n$-ary inputs and outputs.
Applications to CNOT gates, have been discussed in section \ref{VV}.
The general Boolean ring formalism, is translated into the language of classical CNOT gates, through proposition \ref{111}.
\item 
{\bf The Boolean ring   of coherent spaces:}
The set ${\cal L} _{\rm coh}$ of all coherent spaces with the logical OR and AND operations, is a distributive lattice isomorphic to  ${\cal L}$. It is also a Boolean ring, with the operations XOR and AND.
Application to the quantum CNOT gate with coherent states, has been discussed in sections \ref{GG1}, \ref{GG2}.
The non-orthogonal nature of the coherent states, is taken into account with the matrices $g,G$.
The general Boolean ring formalism, is translated into the language of quantum CNOT gates, through proposition \ref{111q}.
\end{itemize}
The work generalizes the concept of coherence, and provides the theoretical foundation for quantum gates with coherent states.

\newpage
\begin{table}
\caption{Outputs of the classical OR, AND, XOR gates with $R=\{A_1, A_2\}$ and the notation in Eq.(\ref{or}).}
\def\arraystretch{2}
\begin{tabular}{|c|c|c|c|c|c|c|c|c|c|c|c|c|c|c|c|c|}\hline
{\rm in} &$(0,0)$&$(1,0)$&$(2,0)$&$(3,0)$&$(0,1)$&$(1,1)$&$(2,1)$&$(3,1)$&$(0,2)$&$(1,2)$&$(2,2)$&$(3,2)$&$(0,3)$&$(1,3)$&$(2,3)$&$(3,3)$\\\hline
{\rm OR}&$0$&$1$&$2$&$3$&$1$&$1$&$3$&$3$&$2$&$3$&$2$&$3$&$3$&$3$&$3$&$3$\\\hline
{\rm AND}&$0$&$0$&$0$&$0$&$0$&$1$&$0$&$1$&$0$&$0$&$2$&$2$&$0$&$1$&$2$&$3$\\\hline
{\rm XOR}&$0$&$1$&$2$&$3$&$1$&$0$&$3$&$2$&$2$&$3$&$0$&$1$&$3$&$2$&$1$&$0$\\\hline
\end{tabular} \label{t1}
\end{table}
\begin{table}
\caption{The (control, target) at the input and output of the classical CNOT gate with $R=\{A_1\}$ and the notation in Eq.(\ref{oror}).}
\def\arraystretch{2}
\begin{tabular}{|c|c|c|c|c|}\hline
{\rm in} &$(0,0)$&$(0,1)$&$(1,0)$&$(1,1)$\\\hline
{\rm out}&$(0,0)$&$(0,1)$&$(1,1)$&$(1,0)$\\\hline
\end{tabular} \label{t2}
\end{table}
\begin{table}
\caption{The (control, target) at the input and output of the classical CNOT gate with $R=\{A_1, A_2\}$ and the notation in Eq.(\ref{or}).}
\def\arraystretch{2}
\begin{tabular}{|c|c|c|c|c|c|c|c|c|c|c|c|c|c|c|c|c|}\hline
{\rm in} &$(0,0)$&$(0,1)$&$(0,2)$&$(0,3)$&$(1,0)$&$(1,1)$&$(1,2)$&$(1,3)$&$(2,0)$&$(2,1)$&$(2,2)$&$(2,3)$&$(3,0)$&$(3,1)$&$(3,2)$&$(3,3)$\\\hline
{\rm out}&$(0,0)$&$(0,1)$&$(0,2)$&$(0,3)$&$(1,1)$&$(1,0)$&$(1,3)$&$(1,2)$&$(2,2)$&$(2,3)$&$(2,0)$&$(2,1)$&$(3,3)$&$(3,2)$&$(3,1)$&$(3,0)$\\\hline
\end{tabular} \label{t3}
\end{table}

\end{document}